\documentclass[nohyperref]{article}

\usepackage{microtype}
\usepackage{graphicx}
\usepackage{subfigure}
\usepackage{booktabs} 
\usepackage{wrapfig}
\usepackage{tikz}
\usepackage{tikz-cd}
\usepackage{adjustbox}

\newcommand{\Ffeat}{\mathcal{F}_{\mathrm{feat}}}
\newcommand{\Fpool}{\mathcal{F}_{\mathrm{pool}}}

\newcommand{\Wfeat}{W_{\mathrm{feat}}}

\usepackage{hyperref}

\newcommand{\angstrom}{\textup{\AA}}
\usepackage{amsmath, amssymb, mathrsfs}
\usepackage{amssymb, color}
\usepackage{amsthm}
\usepackage{booktabs}
\usepackage{multirow,multicol}
\usepackage{subfiles}
\usepackage{soul}

\usepackage{amsmath,amsfonts,bm}

\def\eqref#1{equation~\ref{#1}}

\def\1{\bm{1}}

\def\va{{\bm{a}}}
\def\vb{{\bm{b}}}
\def\vc{{\bm{c}}}

\def\ve{{\bm{e}}}
\def\vf{{\bm{f}}}

\def\vh{{\bm{h}}}

\def\vn{{\bm{n}}}

\def\vu{{\bm{u}}}
\def\vv{{\bm{v}}}

\def\vx{{\bm{x}}}

\def\mH{{\bm{H}}}

\def\mV{{\bm{V}}}

\def\mX{{\bm{X}}}

\DeclareMathAlphabet{\mathsfit}{\encodingdefault}{\sfdefault}{m}{sl}
\SetMathAlphabet{\mathsfit}{bold}{\encodingdefault}{\sfdefault}{bx}{n}

\def\norm#1{\Vert#1\Vert}

\usepackage[accepted]{icml2022}

\usepackage{mathtools}
\usepackage{wrapfig}

\usepackage[capitalize,noabbrev]{cleveref}

\theoremstyle{plain}
\newtheorem{theorem}{Theorem}[section]
\newtheorem{proposition}[theorem]{Proposition}
\newtheorem{lemma}[theorem]{Lemma}

\theoremstyle{definition}
\newtheorem{definition}[theorem]{Definition}

\theoremstyle{remark}
\newtheorem{remark}[theorem]{Remark}

\usepackage[textsize=tiny]{todonotes}

\icmltitlerunning{SE(3) Equivariant Graph Neural Networks with Complete Local Frames}

\begin{document}

\twocolumn[
\icmltitle{SE(3) Equivariant Graph Neural Networks with Complete Local Frames}



\icmlsetsymbol{equal}{*}
\icmlsetsymbol{a}{\dag}
\begin{icmlauthorlist}
\icmlauthor{Weitao Du}{equal,a,comp}
\icmlauthor{He Zhang}{equal,a,sch}
\icmlauthor{Yuanqi Du}{a,go}
\icmlauthor{Qi Meng}{yyy}
\icmlauthor{Wei Chen}{a,comp}
\icmlauthor{Nanning Zheng}{sch}
\icmlauthor{Bin Shao}{yyy}
\icmlauthor{Tie-Yan Liu}{yyy}
\end{icmlauthorlist}

\icmlaffiliation{yyy}{Microsoft Research, USA. \dag \, denotes that the work was done when the authors were visiting Microsoft Research}
\icmlaffiliation{comp}{Chinese Academy of Sciences, China}
\icmlaffiliation{sch}{Xi'an Jiaotong University, China}
\icmlaffiliation{go}{George Mason University, USA}
\icmlcorrespondingauthor{Qi Meng}{meq@microsoft.com}
\icmlcorrespondingauthor{Weitao Du}{duweitao@amss.ac.cn}

\icmlkeywords{Machine Learning, Equivariance, Graph neural networks}

\vskip 0.3in
]



\printAffiliationsAndNotice{\icmlEqualContribution} 

\begin{abstract}
Group equivariance (e.g. SE(3) equivariance) is a critical physical symmetry in science, from classical and quantum physics to computational biology. It enables robust and accurate prediction under arbitrary reference transformations. In light of this, great efforts have been put on encoding this symmetry into deep neural networks, which has been shown to improve the generalization performance and data efficiency for downstream tasks. 
Constructing an equivariant neural network generally brings high computational costs to ensure expressiveness. Therefore, how to better trade-off the expressiveness and computational efficiency plays a core role in the design of the equivariant deep learning models. 
In this paper, we propose a framework to construct SE(3) equivariant graph neural networks that can approximate the geometric quantities efficiently.
Inspired by differential geometry and physics, we introduce equivariant local complete frames to graph neural networks, such that tensor information at given orders can be projected onto the frames. The local frame is constructed to form an orthonormal basis that avoids direction degeneration and ensure completeness. Since the frames are built only by cross product operations, our method is computationally efficient. 
We evaluate our method on two tasks: Newton mechanics modeling and equilibrium molecule conformation generation. Extensive experimental results demonstrate that our model achieves the best or competitive performance in two types of datasets (Code will be released soon). 
\end{abstract}

\section{Introduction}
The success of CNNs \citep{krizhevsky2012imagenet, he2016deep}  and its SO(2) group extension \citep{cohen2016group} are sufficient to justify the benefits of explicitly translationally
equivariant or SE(2) equivariant neural network architectures. On the other hand, Graph Neural Networks (GNNs), which are superior to modeling the high-dimensional structured data with permutation equivariance, 
bring a new opportunity to model physical systems (especially the various many-body systems) in an end-to-end manner. However, since 3D many-body systems follow other constraints like SE(3) symmetry, pure black-box GNN models show limitations on generalization in this scenario and symmetry-preserving GNN models have become a hot research direction.

The main challenge to be solved for developing general equivariant GNN models is how to express tensors and conduct nonlinear operations in a reference-free way. To represent and manipulate equivariant tensors of arbitrary orders, some approaches resort to equivariant function spaces such as spherical harmonics \citep{thomas2018tensor, fuchs2020se,pmlr-v119-bogatskiy20a,fuchs2021iterative} or lifting the spatial space to high-dimensional spaces such as Lie group space \citep{cohen2016group,cohen2018spherical, cohen2019general,finzi2020generalizing,hutchinson2021lietransformer}. Since no restriction on the order of tensors is imposed on these methods, sufficient expressive power of these models is guaranteed. 
Unfortunately, transforming a many-body system into those high-dimensional spaces or calculating equivariant functions (e.g. irreducible representations and tensor product decomposition) usually brings excessive computational cost and optimization difficulty, which is unacceptable in some real-world scenarios. 
On the other hand, there are also equivariant neural networks \cite{schutt2017schnet, satorras2021en, kohler2019equivariant} 
that directly implement equivariant operations in the original space, 
providing an efficient way to preserve equivariance without complex 
equivariant embeddings. However, most of these models only use radial direction as incomplete frames along with embedding of scalar features and abandon higher-order tensor inputs. Thus, they face the direction degeneration problem and are insufficient for expressing more complex geometric quantities like torsion force (See Section \ref{motivation}).
In light of this, how to better trade-off the expressiveness and computational efficiency plays a core role in the design of equivariant models.

In this paper, we propose a framework to construct SE(3) equivariant network with \textbf{c}omplete \textbf{lo}cal \textbf{f}rames, called \textbf{ClofNet} that can faithfully and efficiently approximate the geometric quantities. Inspired by differential geometry and physics, we incorporate the equivariant local complete frames into graph neural networks, where complete means each local frame forms an orthogonal basis (reference frame) in the 3D vector space with no direction degeneration. 
To faithfully express the geometric quantities (e.g., scalars, position vectors), we introduce a scalarization block to project them onto our local frames, transforming the tensors into the scalarized coefficients. Afterward, the coefficients are fed into a flexible graph neural network for message passing.
The group equivariance of the whole model is guaranteed due to the equivariance of the frames and the expressiveness is guaranteed due to the completeness of the frames. Moreover, the construction of frames only involves cross product operations,
making our method computationally efficient (e.g., compared with the Clebsch-Gordan tensor product).
Finally, the method to construct the frames is extended to scenarios where the input contains high-order tensors in a similar way. We evaluate the proposed framework on two many-body scenarios that require equivariance with extensive ablation studies: (1) the synthetic many-body system dynamics modeling and (2) the real-world molecular conformation generation. Our model achieves the best performance or competitive results in various types of scenarios. Especially, the proposed ClofNet still keeps competitive performance even with $0.4\%$ training samples on the many-body system, showing its superiority on sample complexity. 

\section{Preliminaries}
In this section, we set up the necessary mathematical preliminaries for understanding the follow-up sections.
Let $\mX=(\vx_1,\dots,\vx_N) \in \mathbb{R}^{N\times 3}$ be a many-body system living in $\mathbb{R}^3$, where $N$ is the number of particles. For particle $i$, we use $\vx_i(t)$ to denote its position at time $t$ and
define its neighborhood particles as  $\mathcal{N}(\vx_i(t))$. 
Let $\times$ denote the cross product of two vectors and $\otimes$ denote the tensor product.   

\paragraph{SE(3) group and Equivariance. }
In the Euclidean space $\mathbb{R}^3$, 
we consider affine transformations that preserve the distance between any two points, i.e., the isometric group SE(3). We call it the symmetry group w.r.t. the Euclidean metric, and it turns out that SE(3) can be generated by the translation group and the 3D rotation group SO(3) (See rigorous definition in Appendix). 

Once given a symmetry group, it's valid to define quantities that are ``equivariant'' under the symmetry group. 
Given a function $f: \mathbb{R}^m \rightarrow \mathbb{R}^n$, assuming the symmetry group $G$ acts on $\mathbb{R}^m$ and $\mathbb{R}^n$ and we denote the actions by $T_g$ and $S_g$ respectively, then $f$ is $G$-\textbf{equivariant} if
$$f(T_g x) = S_g f(x),\ \ \ \forall x \in \mathbb{R}^m\ \text{and}\ g \in G.$$
For SO(3) group, if {$n=1$, i.e., the output of $f$ is a scalar, then the group action on $\mathbb{R}^1$ is the identity map, in this case $f$ should be $SO(3)$-invariant \citep{thomas2018tensor}: $f(T_gx) = f(x).$}

The geometric tensors that satisfy the group equivariance are formally defined as tensor field.  
{Let {\small$\{\frac{\partial}{\partial x_{i}}\}_{i=1}^3$} and {\small$\{dx^i\}_{i=1}^3$} be the tangent vectors and dual vectors in $\mathbb{R}^3$ respectively}. Then recall the definition of tensor field on $\mathbb{R}^3$ w.r.t. the SO(3) group:
\begin{definition} \label{ten} [Definition 2.1.10 of \citep{jost2008riemannian}]
An \textbf{(r,\,s)-tensor field} $\theta$ is a multi-linear map from a collection of r dual vectors and s vectors in $\mathbb{R}^3$ to $\mathbb{R}$:
{\small$\theta(x) = \theta^{i_1 \cdots i_r}_{j_1 \cdots j_s}\  \frac{\partial}{\partial x_{i_1}}\otimes \cdots \otimes \frac{\partial}{\partial x_{i_r}}\otimes dx^{j_1} \otimes \cdots \otimes dx^{j_s}.$}
It implies that under SO(3) transformation $g :=\{g_{ij}\}_{1 \leq i,j \leq 3}$,  the tensor field $\theta$ transforms equivariantly:
{\small$\theta^{i_1' \cdots i_r'}_{j_1' \cdots j_s'} = g_{i_1'i_1}\cdots g_{i_r'i_r}g^T_{j_1j_1'}\cdots g^T_{j_sj_s'} \theta^{i_1 \cdots i_r}_{j_1 \cdots j_s},$}
where $g^T$ is the inverse of $g$.
\end{definition}
{The notion of tensor field is introduced for expressing and embedding geometric and physical quantities, under which the representation satisfies the group equivariance. The tensor or geometric tensor in this paper refers to the tensor field defined in Definition~\ref{ten}. }
A canonical example of $(1,0)$-type tensor fields is the differential(velocity) of a dynamical system $\mX(t)$, evolving w.r.t. time $t$. We name it an  \emph{equivariant vector field} (see Appendix \ref{evo}).  

\paragraph{Graph neural networks. }Since many body system modeling is independent of the labeling on the observed particles (permutation equivariance), we will use graph-based message-passing mechanism \citep{gilmer2020message} in this paper. Given a graph $G = (V(G),E(G))$,
let $h_i$, $x_i$ and $e_{ij}$ be the node features, node positions and edge attributes, respectively.
EGNN \citep{satorras2021en} provides an efficient way to learn equivariant representations, which updates edge message $m_{ij}$, node embeddings $h_i$ and coordinates $x_i$ as:
\begin{align}
\label{eq:egnnmessages}
m_{ij} &= \phi_m(h_i, h_j, \lVert \vx_j - \vx_i \rVert^2, e_{ij}), \nonumber\\ 
\vx_{i} &= \vx_i + C\sum_{j \neq i}(\vx_i - \vx_j)\phi_x(m_{ij}) \\ \nonumber
h_{i} &= \phi_h(h_i, \sum_{j \in \mathcal{N}(i)}m_{ij})
\end{align}
where $\phi_m$, $\phi_x$ and $\phi_h$ denote three neural networks. $C$ is the normalization factor.
 
\section{Motivation from many body interactions. } \label{motivation}
 In this section, we provide an example to motivate our methodology. As shown in Figure \ref{fig:motivation}, we consider a two-body conservative sub-system $(\vx_i,\vx_j)$ in many-body systems with no time dependency, which can be seen as a graph with two nodes and one edge. Then the force on each particle is the gradient of the potential energy. In general, the groundtruth potential energy function of the two particles is decomposed into two parts: $U(\vx_i,\vx_j) := U_1(\norm{\vx_i-\vx_j}) + U_{\text{ext}}(\vx_i,\vx_j)$, where the first part is a function of the relative distance $\norm{\vx_i-\vx_j}$ between particles (e.g., electrostatic force), the second part describes the influence of external force fields that may depend on both relative distance and angles (e.g. electromagnetic field, torsion force (\ref{torsion})). Therefore, $\nabla U_{\text{ext}}(\vx_i,\vx_j)$ can be along arbitrary directions in 3D space besides the radial direction $\vx_i - \vx_j$. Most mainstream approaches in equivariant neural network tackle such force prediction problem by taking invariant features (e.g., relative distance) as input and expressing the force along the radial direction, i.e., $\hat{F_i}=k \cdot (\vx_i-\vx_j)$ \citep{schutt2017schnet, satorras2021en}, which is not sufficient when the gradient $\nabla U_{\text{ext}}(\vx_i,\vx_j)$ is not along radial direction.   To remedy this issue, we put forward an orthogonal SO(3) equivariant frame (basis) for this two-body system.  As shown in Figure \ref{fig:motivation}, we take $\vx_i - \vx_j$ as the first component/direction $\va$, and then include $\vx_i \times \vx_j$ as the second direction $\vb$. Finally, the cross product of the first two directions is taken as the third direction $\vc$. Using such equivariant frame, any force directions can be equivariantly expressed \textbf{with no direction degeneration} by decomposing itself into the three orthogonal directions. The theoretical and technical details about our augmented frames will be discussed later.

\begin{figure}
    \centering
    \includegraphics[width=0.7\linewidth]{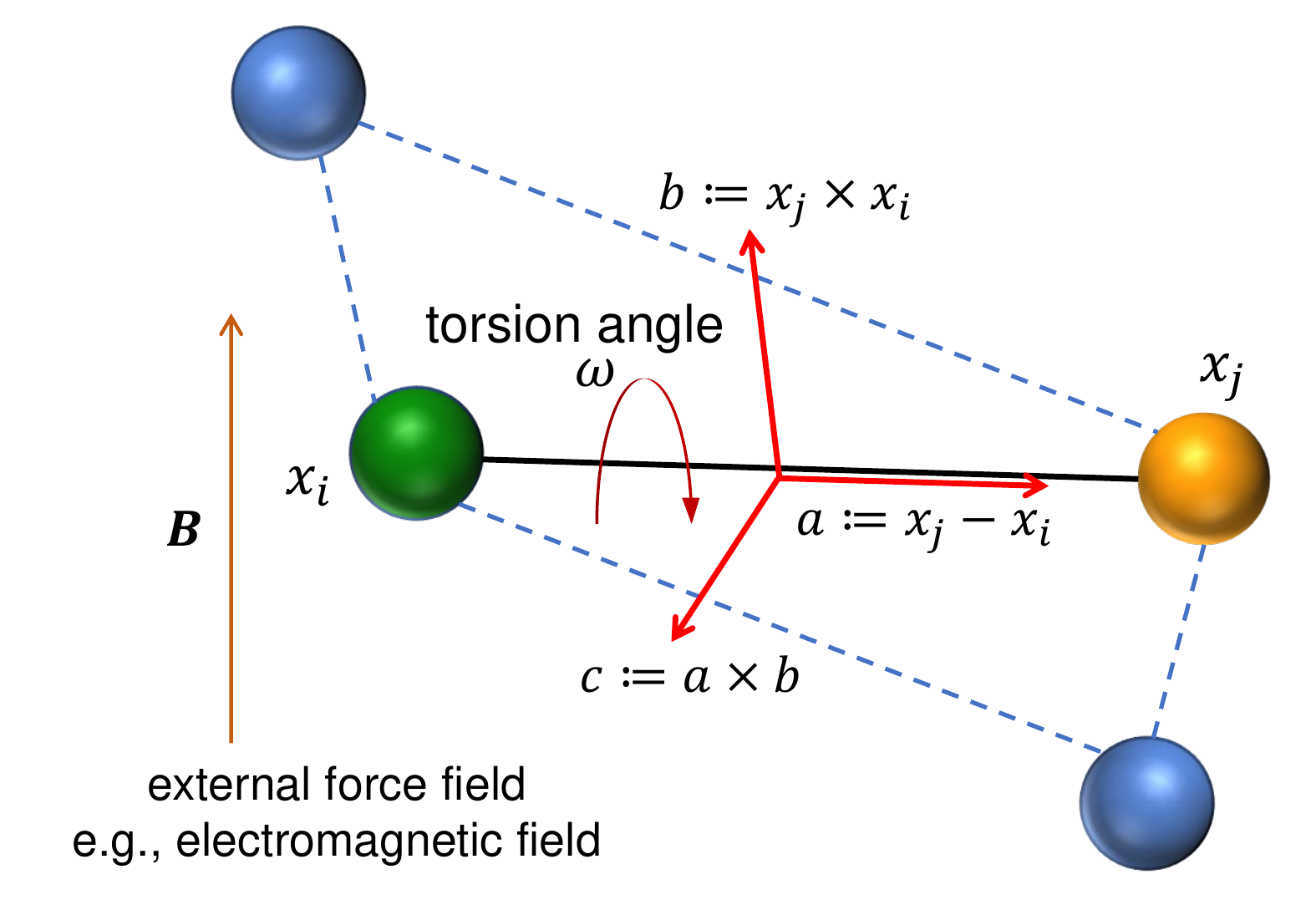}
    \caption{Example of a two-body system to motivate our complete local frame.}
    \label{fig:motivation}
    \vskip -0.2in
\end{figure}

\begin{figure*}
    \centering
    \includegraphics[width=0.8\linewidth]{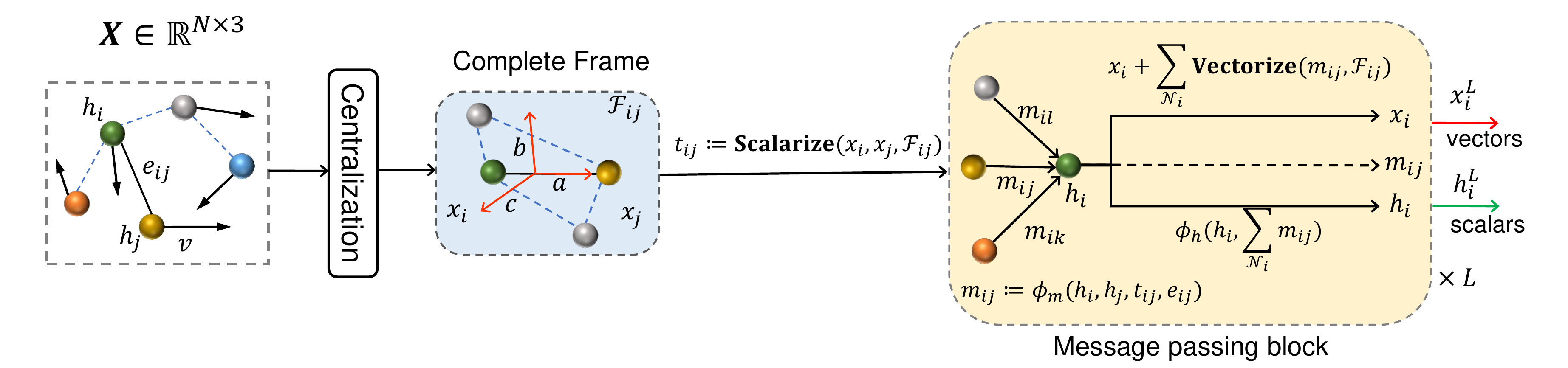}
    \caption{An overview of ClofNet. For a many-body system $\mX$, we first centralize the positions to preserve translation equivariance. Then we introduce a tuple of edge-level complete frame $\mathcal{F}_{ij}$ to transform the geometric tensors $\vx_i$ into SO(3)-invariant scalars. Afterwards, the scalar embeddings $t_{ij}$, pre-defined node features $h_i$ and edge features $e_{ij}$ are fed to the graph message-passing block to learn edgewise embeddings $m_{ij}$. Finally, a vectorization block transforms the edgewise embeddings into nodewise vector fields $\vx_i^L$ or scalar fields $h_i^L$.}
    \label{fig:framework}
\end{figure*}
\section{Methodology}
In this section, we provide a detailed description of ClofNet.  
To preserve the physical symmetries of the system, we represent the system as a spatial graph and construct \textbf{ClofNet} based on it with three key components: (1) a \textbf{Scalarization} block to project the geometric tensors onto our constructed equivariant frames and concatenate all coefficients as the SO(3)-invariant scalar representations attached to each node; (2) a \textbf{Graph message-passing} block to learn SO(3)-invariant edgewise embeddings by propagating and aggregating the scalars on the graph; and (3) a \textbf{Vectorization} block to reverse scalar representations back to tensors. A brief overview of ClofNet is illustrated in Figure \ref{fig:framework}. 
The translation symmetry can be easily realized by moving the particle system's centroid at $t=0$ to the origin (the \textbf{Centralization} operation in Figure \ref{fig:framework}). Permutation equivariance automatically holds for the message-passing block. We provide detailed proofs about these symmetries in Appendix \ref{symmetry}. Now we concentrate on SO(3) symmetry. 

\paragraph{Scalarization Block}
The scalarization block is designed to transform geometric tensors into edgewise SO(3)-invariant scalar coefficients/features by introducing a novel tuple of local frames. 
Consider a particle pair $(\vx_i(t), \vx_j(t) )$ at time $t$, where $\vx_j(t) \in \mathcal{N}(\vx_i(t))$.
We define the edgewise SO(3)-invariant scalars as
$s_{ij} := \mathrm{Scalarize} (\vx_i(t), \vx_j(t),\mathcal{F}_{ij})$, where $\mathrm{Scalarize}$ is the scalarization operation under local equivariant frames $\mathcal{F}_{ij}$ defined below.  

\textbf{(1) Local equivariant frame construction. } Following the last section, based on the particle pair $(\vx_i(t), \vx_j(t) )$,
let $\va_{ij}^t = \frac{\vx_i(t) - \vx_j(t)}{\norm{\vx_i(t) - \vx_j(t)}}$. Define
\begin{equation} \label{basis}
\vb_{ij}^t = \frac{\vx_i(t) \times \vx_j(t)}{\norm{\vx_i(t) \times \vx_j(t)}}\ \ \text{and}\ \ \ \vc_{ij}^t = \va_{ij}^t \times \vb_{ij}^t,\end{equation}
Then we build an SO(3)-equivariant frame $\mathcal{F}_{ij}:= \mathrm{EquiFrame}(\vx_i, \vx_j)=(\va_{ij}^t,\vb_{ij}^t,\vc_{ij}^t)$. In practice we add a small constant $\epsilon$ to the normalization factor in case that $\vx_i$ and $\vx_j$ collapse.
If the $3 \times 3$ matrix $(\va_{ij}^t,\vb_{ij}^t,\vc_{ij}^t)$ is non-degenerate, then $\mathcal{F}_{ij}$ is complete in the sense that it formulates a \emph{local orthonormal basis}, where `local' means the frame is a relative reference frame on particles located at the tangent space of $\vx_i(t)$. Note that
these frames are permutation-equivariant among particles.
The extreme event that all $\mathcal{F}_{ij}$ ($j\in\mathcal{N}(i)$) degenerate for node $i$ only happens when all particles are restricted to a straight line, which is a measure zero set in $\mathbb{R}^3$. Therefore, we assume $\mathcal{F}_{ij}$ is non-degenerate from now on. Proof for SO(3)-equivariance of $\mathcal{F}_{ij}$ is provided in Proposition \ref{exp}.    

\textbf{(2) Invariant scalarization of geometric tensors. }  \label{Invariant scalarization}
With the complete equivariant frame, we realize the scalarization operation \citep{kobayashi1963foundations} as follows. First of all, we project the input vector of particle $\vx_i$ onto the frame $\mathcal{F}_{ij}$ ($(\va_{ij}^t,\vb_{ij}^t,\vc_{ij}^t)$), and the corresponding \emph{coefficients} are obtained via inner product as follows: 
\begin{equation} \label{sca}
    (\vx_i \cdot \va_{ij}^t, \vx_i \cdot \vb_{ij}^t, \vx_i \cdot \vc_{ij}^t).
\end{equation}
Here we demonstrate that the set of obtained coefficients is actually a SO(3)-invariant scalars tuple and define the process obtaining such scalars as $\mathrm{Scalarize}$ operation.
Let $g \in SO(3)$ be an arbitrary orthogonal transformation, then $\vx_i \rightarrow g \vx_i$, and
$(\va_{ij}^t,\vb_{ij}^t,\vc_{ij}^t) \rightarrow (g\va_{ij}^t,g \vb_{ij}^t,g \vc_{ij}^t).$
We can derive that Equation (\ref{sca}) undergoes
\begin{align}
(\vx_i \cdot &\va_{ij}^t, \vx_i \cdot \vb_{ij}^t, \vx_i \cdot \vc_{ij}^t) \nonumber\\
&\rightarrow ((\vx_i)^T g^T g\va_{ij}^t, (\vx_i)^T g^T g\vb_{ij}^t, (\vx_i)^T g^T g\vc_{ij}^t) \nonumber \\ \label{new} &= (\vx_i \cdot \va_{ij}^t, \vx_i \cdot \vb_{ij}^t, \vx_i \cdot \vc_{ij}^t),    
\end{align}
where we use the fact the transpose of $g$ satisfies $g^T g = I$ ($I$ is the identity matrix), according to the definition of SO(3) group. 
Next, we discuss that the $\mathrm{Scalarize}$ operation could be extended to transform \textbf{high-order geometric tensors} into SO(3)-invariant scalars.
Suppose the input data contain (2,0)-type tensors, then by extending the each local complete frame $(\va_{ij},\vb_{ij},\vc_{ij})$ through tensor product, it's easy to check that
\begin{align} 
\label{tensor product bases}
\begin{aligned}
&\{\va_{ij} \otimes \va_{ij}, \vb_{ij} \otimes \vb_{ij},\vc_{ij} \otimes \vc_{ij}, \va_{ij} \otimes \vb_{ij},\\
&\vb_{ij} \otimes \va_{ij},\va_{ij} \otimes \vc_{ij},\vc_{ij} \otimes \va_{ij},\vb_{ij} \otimes \vc_{ij},\vc_{ij} \otimes \vb_{ij}\}
\end{aligned}
\end{align}
forms an equivariant (the equivariance of the tensor products under the equivariant frame are given by Def \ref{ten}) orthonormal frame of the (2,\,0)-type tensor space. Then the scalarization of a tensor is just to concatenate the \emph{coefficients} of the linear expansion under this frame. In the same way as (\ref{new}), we can prove that the coefficients are also SO(3)-invariant scalars. Given a (2,\,0)-symmetric tensor $\pmb{\theta}$ (e.g., energy-momentum tensor), under the complete frame $\mathcal{F}_{ij} =(\va,\vb,\vc)$, $\pmb{\theta}$ can be expressed as:
\begin{align}
  \pmb{\theta} =& \theta^{aa}\va \otimes \va + \theta^{bb}\vb \otimes \vb +\theta^{cc}\vc \otimes \vc  + \theta^{ab}(\va \otimes \vb + \vb \otimes \va) \nonumber \\ 
  \label{scalar} & + \theta^{ac}(\va \otimes \vc + \vc \otimes \va) +  \theta^{bc}(\vb \otimes \vc + \vc \otimes \vb).
\end{align}
The scalars tuple $s: = \{\theta^{aa}, \theta^{ab}, \dots\}$ are the scalarization of $\pmb{\theta}$ under the equivariant frame $\mathcal{F}_{ij}$, which are SO(3)-invariant and will be discussed in detail in Section 5. Let $X$ be the original linear space of all tensors, $\mathcal{S}(X)$ be the space of scalars tuples. Then one of the important benefits of scalarization under an equivariant frame is the scalarization transforms $X$ to $\mathcal{S}(X)$ in an \textbf{invertible} way. Given a scalars tuple $s := \{x^a,x^b,x^c\}$ that represents our scalarized input of some position vector $\vx$ under the equivariant frame $(\va,\vb,\vc)$, then the inverse from $s$ to vector $\vx$ is given by 
\begin{equation} \label{invert}
\vx \equiv x^a \va + x^b \vb + x^c \vc.\end{equation} 
In our experiments, most of the geometric input is  position vectors $\{\vx_i\}_{i=1}^n$, i.e., $(1, 0)$-type tensors. For the simplicity of description, we denote all other vector features (e.g., velocity) as $\pmb{\chi}_i \in \mathbb{R}^{m \times 3}$, where $m$ is the number of potential vector types.
The $\mathrm{Scalarize}$ operation is defined as: $\mathrm{Scalarize}(\vx_i,\, \vx_j,\, \pmb{\chi}_i, \pmb{\chi}_j,  \mathcal{F}_{ij}) = (\vx_i \cdot \va_{ij}^t,\, \vx_i \cdot \vb_{ij}^t,
\vx_i \cdot \vc_{ij}^t,\, \vx_j \cdot \va_{ij}^t, \vx_j \cdot \vb_{ij}^t,\, \vx_j \cdot \vc_{ij}^t, \dots)$.

\paragraph{Graph message-passing Block} \label{message-passing}
After encoding the geometric tensors into SO(3)-invariant scalars $s_{ij}$, we first embed them alone with other pre-defined node/edge attributes ($h_j$, $e_{ij}$) into high-dimensional representations, and leverage a graph message-passing (GMP) block 
to learn the SO(3)-invariant edgewise message embeddings $m_{ij}$ by propagating and aggregating information on the graph $\mathcal{G}_X$. Since the block is processing on invariant scalars, ClofNet is \textbf{flexible} in the sense that any nonlinear activations can be applied when performing message-passing such that the outputs are still invariant scalars. For more complex tasks, our model is capable of including the attention mechanism. We provide further design details in Appendix \ref{app:ga_b}.
\begin{algorithm}[tb]
   \caption{ClofNet}
   \label{alg:evfn}
\begin{algorithmic}[1]
   \STATE {\bfseries Input:} $\mathbf{X}=(\vx_1,\dots,\vx_N) \in \mathbb{R}^{N\times 3}$, $h_i \in \mathbb{R}^{h}$, $e_{ij} \in \mathbb{R}^{e}$, $\pmb{\chi}_i \in \mathbb{R}^{m \times 3}$, graph $\mathcal{G}_X$
   \STATE Initialize $\textbf{X}^1 \leftarrow \mathbf{Centralize}(\textbf{X})$.
   \STATE $\mathcal{F}_{ij}^1=\mathbf{EquiFrame}(\vx_j^1,\, \vx_j^1)$; 
   \STATE $s_{ij}=\mathbf{Scalarize}(\vx_i^1,\, \vx_j^1, \pmb{\chi}_i, \pmb{\chi}_j,\mathcal{F}_{ij}^1)$;
   \FOR{($l=1; l < L; l++$)}
    \STATE $m_{ij}^l = \phi_m^l(s_{ij}, h_i^l, h_j^l, e_{ij})$;
    \STATE $\mathbf{h}_i^{l+1} = \phi_h^l (\mathbf{h}_i^{l},\sum_{j \in \mathcal{N}(i)} m_{ij}^l)$;
    \STATE $\mathcal{F}_{ij}^{l+1}=\mathbf{EquiFrame}(\vx_i^l,\, \vx_j^l)$; 
     \STATE $\vx_i^{l+1} = \vx_i^{l} + \frac{1}{N}\sum_{j \in \mathcal{N}(\vx_i)}{\mathbf{Vectorize}(m_{ij}^l,\, \mathcal{F}_{ij}^l)}$;
   \ENDFOR
   \STATE {\bfseries Output: $x_i^L$ or $h_i^L$}
\end{algorithmic}
\end{algorithm}

\paragraph{Vectorization block} 
This block is necessary only for SO(3)-equivariant vector output (or general tensor outputs). For scalar output, a simple pooling layer following the graph message-passing block is sufficient. Given the propagated edgewise message $m_{ij}$, the vectorization block is designed to transform these scalars back to equivariant vectors (inverse of scalarization), which requires pairing $m_{ij}$ with the corresponding complete frame.
More precisely, we first project $m_{ij}$ into a scalar triple $\{x_{ij}^a,x_{ij}^b,x_{ij}^c\}$, then define the vectorization process as:
    \begin{equation} \label{vector}
    (x_{ij}^a,x_{ij}^b,x_{ij}^c) \xrightarrow{\text{Pairing}} x_{ij}^a\va_{ij} + x_{ij}^b\vb_{ij} + x_{ij}^c\vc_{ij}.\end{equation}
We encapsulate the pairing process as $V_{ij} = \mathrm{Vectorize}(m_{ij},\mathcal{F}_{ij})$. 
Finally, we aggregate all vectors $V_{ij}$ associated with $\vx_i$ to estimate the ground-truth vector field $V_i$. Since the local frames are SO(3) equivariant, $V_i$ is also SO(3) equivariant. This method is also applicable to other types of tensors by pairing with the corresponding tensor product (\ref{tensor product bases}) of the equivariant frame.

In conclusion, we give a graphic illustration on learning an SO(3)-equivariant vector $\vv(\vx_1, \vx_2)$, with vector input $\vx_1$, $\vx_2$. Under the equivariant frame of $(\vx_1, \vx_2)$, we build ClofNet as a commutative diagram:
\begin{tikzcd}[column sep=scriptsize, row sep=scriptsize] \label{commutative}
     \vx_1, \vx_2 \arrow[r,  "\vf"] \arrow[d, "\text{Scalarize}"]
      & \vv(\vx_1, \vx_2)  \\
     s_1,s_2 \arrow[r, "\Tilde{f}"]
      & \Tilde{f}(s_1,s_2) \arrow[u,  "\text{Vectorize}"'],
\end{tikzcd}
where $\Tilde{f}$ is our learning model parameterized by a non-linear neural network, and $s_i \in \mathcal{S}(X)$ is the invariant scalar tuple corresponding to the equivariant vector $\vx_i$. Then by definition, we realize $\vf: \mathbb{R}^3 \times \mathbb{R}^3 \rightarrow \mathbb{R}^3$:
\begin{equation}\label{eq:equi-f}
  \vf = \text{Vectorize} \circ \Tilde{f} \circ \text{Scalarize},
\end{equation}
which is proved to be SO(3)-equivariant in next section. So far, we have achieved permutation and SE(3) equivariance by employing the centralization and the three blocks. 
The workflow of our method is summarized in Algorithm \ref{alg:evfn}. Note that we can leverage the scalarization block in each message-passing block by inserting the $\mathrm{Scalarize}$ operation after line 9 in Algorithm \ref{alg:evfn}. Empirically, we find no significant gains using this modification in our two experiments. Therefore we exclude it to reduce the computation cost. The computational complexity (scalability) of our algorithm is given in appendix \ref{SCALABILITY}.

\section{Theoretical properties of ClofNet}
\label{draft:theory}
\paragraph{Equivariance guarantee} First, the translation invariance and equivariance are guaranteed by moving the center of a many-body system. The $SO(3)$ group equivariance of ClofNet can be obtained by the following Proposition.
\begin{proposition}
\label{Equivariance}
(1) The local complete frames $(\va_{ij}^t,\vb_{ij}^t,\vc_{ij}^t), $ $\forall i,j$ defined by Equation~(\ref{basis}) is equivariant under SO(3) transformation of the spatial space. 

(2) The Scalarize operation, i.e., $\mathrm{Scalarize}(\vx_i,\vx_j,\mathcal{F}_{ij})$ is invariant under SO(3) transformation.
\end{proposition}

The proof for the above proposition is put into Appendix \ref{app:mini}, \ref{equi}. Using this proposition, we conclude that $\vf = \mathrm{Vectorize} \circ \Tilde{f} \circ \mathrm{Scalarize}$ defined in Equation~(\ref{eq:equi-f}) is equivariant under SO(3) transformation.

\paragraph{Expressiveness of ClofNet}
The expressive power of ClofNet can be decomposed into two parts: the scalarization (and vectorization) part and the ordinary Graph Neural Network part. By the \textbf{invertibility} of (\ref{invert}), there exists a bijection between a tensor $\vx$ and its scalarization $\mathcal{S}(\vx)$. It implies that the scalars can fully characterize the input tensors losslessly (an illustration in terms of mutual information can be found in Appendix \ref{Information loss}). Then the expressive power lies in the expressiveness of the flexible graph message-passing block. The following Theorem verifies the universal approximation property of ClofNet if the equipped graph message-passing block is expressive. See more details and proof in Appendix~\ref{univer}.
\begin{theorem} \label{universal approximation}
If the graph message-passing block is selected to be a neural network that can approximate permutation-equivariant polynomials on the graph (e.g. the tensorized graph neural network \citep{maron2018invariant} or the minimal universal architecture \cite{dym2020universality}), and a linear fully-connected layer connects the output of the graph message-passing block with the vectorization block, then ClofNet has the universality in the space of continuous $SE(3)$ and permutation equivariant functions. 
\end{theorem}
\textbf{Remark: } 
(1) The proof relies on the equivalence between direct tensor products ( Equation (\ref{product})) of equivariant tensors and polynomials of scalars obtain by scalarization, which is realized only through our complete frames. 
(2) From the theoretical aspect, expressing all permutation-equivariant polynomials requires expressing multiplication operations of invariant scalars on different nodes, which means the scalarization block should contain the vector features of all nodes simultaneously, not only direct neighbors. In practice, we follow previous works to select the common graph convolutional network or the graph transformer network as our graph message-passing block, which brings more efficiency if the contributions of higher-order hops are negligible. 

\section{Related Work}
Various symmetries are considered to be embedded into the network for different tasks. 
Vanilla CNN models are naturally translation invariant, and more works on 2D-symmetry include \citep{he2021efficient,franzen2021general,dieleman2016exploiting,cohen2016group,cohen2016steerable,weiler2019general,shen2020pdo}.
Another common symmetry for 3D objects is SO(3) group \citep{esteves2019equivariant,esteves2019cross,worrall2018cubenet,kondor2021fully}. 
In terms of methodology, existing equivariant networks 
can be roughly classified into two categories by whether conducting all operations in the original space or not.
The first category of methods lift the data into high-dimensional spaces (e.g., lie group) or introduce equivariant functions (e.g., spherical harmonics) to preserve equivariance \citep{worrall2017harmonic, thomas2018tensor, kondor2018clebsch, weiler2018learning, weiler20183d, weiler2019general, esteves2020spin,romero2020attentive,klicpera2020directional, anderson2019cormorant,batzner2021se}.
They exhibit sufficient expressive power \citep{dym2020universality} but usually bring extra computational costs for calculating spherical harmonic embedding (steerable features \citep{2110.02905}) or performing integration on Lie groups. If we restrict the equivariant function class to be linear, it turns out that group convolution in a message-passing form is the legal implementation \citep{bekkers2019b,kondor2018generalization,2110.02905,anderson2019cormorant}.

Our work follows methods of the second category that operate on the original space in a computationally efficient way \citep{schutt2018schnet,kohler2019equivariant,shi2021learning,deng2021vector}. However, most of these approaches (e.g., EGNN \citep{satorras2021en}) abandon a certain amount of geometric (tensor) information, causing their expressive power to be restricted. 
Recall the EGNN algorithm (\ref{eq:egnnmessages}), although the messages are fully non-linear, only the invariant distance feature $\lVert \vx_j - \vx_i \rVert$ is fed into the model. On the other hand, the success of \citep{ 2110.02905,batzner2021se} has demonstrated the effectiveness for integrating sophisticated geometric and physical features in a non-linear way by equivariant function embedding (or second-hop scalar features \citep{klicpera2020directional}). Although implementing geometric features by spherical harmonics and CG tensor products (tensor product with a decomposition into irreducible representations) are manifestly equivariant, this is not the only way. 
In fact, they are just sub-solutions that satisfy theorem 1 of \citet{he2021efficient}.
Different from the existing methods, we propose a flexible and efficient architecture that avoids complex vector-level transformations while preserving complete geometric information through the equivariant frame (see the discussion at the end of Section \ref{motivation} and Section \ref{draft:theory}). More related works on encoding geometric (steerable) features and their relation with our model can be found in remark \ref{Incorporating}.
 \section{Experiments}
\label{section:exp}
We introduce two many-body scenarios to validate the strength of ClofNet on approximating complex geometric information and avoiding direction degeneration: (1) the simulated Newton mechanics systems and (2) the real-world molecular systems.
\begin{table*}[t]
  \caption{MSE for future position prediction over four datasets. The forward time is measured by averaging over multiple batches on an Nvidia Tesla V100 GPU, based on a batch size of 100 samples.}
  \label{table:nbody}
  \centering
  \begin{tabular}{cccccc}
    \toprule
      Model & ES(5) & ES(20) & G+ES(20) & L+ES(20) & Forward time (s)\\
    \midrule
    GNN & 0.0131 & 0.0720	& 0.0721 & 0.0908 & 0.0026\\
    TFN & 0.0236& 0.0794& 0.0845 & 0.1243 & 0.0247\\
    SE3 Transformer & 0.0329 & 0.1349 & 0.1000& 0.1438 & 0.0392\\
    Radial Field & 0.0207 & 0.0377 & 0.0399 & 0.0779 & 0.0040\\
    EGNN & 0.0079 & 0.0128 & 0.0118 & 0.0368 & 0.0067\\
    ClofNet & \textbf{0.0065} & \textbf{0.0073} & \textbf{0.0072} & \textbf{0.0251} & 0.0096\\
    \bottomrule
  \end{tabular}
  \vskip -0.1in
\end{table*}
\subsection{Newtonian many-body system}
\label{draft:newton}
In this experiment, we leverage ClofNet to predict future positions of synthetic many-body systems that are driven by Newtonian force fields, which is a typical equivariant task because applying any rotations or translations on the initial system state leads to
the same transformations on the future system positions.
This task is inspired by \citep{kipf2018neural,fuchs2020se,satorras2021en}, where a 5-body charged system is controlled by the electrostatic force field. Note that the force direction between any two particles is always along the radial direction in the original setting. To validate the effectiveness of ClofNet on modeling arbitrary force directions, we also impose two external force fields into the original system, a gravity field and a Lorentz-like dynamical force field, which can provide more complex and dynamical force directions.

\paragraph{Dataset} We design four systems for three kinds of force fields. (1) The naive charged 5-body system controlled by the electrostatic force field, where each particle carries a positive or negative charge, with initial position and velocity in a 3-dimensional space, i,e., ES(5). (2) Another electrostatic system consists of 20 charged particles, i.e., ES(20). (3) The 20-body system controlled by both the electrostatic field and the gravity field, i.e., G+ES(20). The gravity field is along the z-axis $\vf_{\eta}=(0,\, 0,\, g)$, where the gravitational acceleration $g$ is set to $0.98$. (4) The 20-body system controlled by both the electrostatic field and the Lorentz-like force field, i.e., L+ES(20). The force field is perpendicular to the direction of each particle's velocity: $\vf_{l}(\vv(t)) = q\vv(t) \times \mathbf{B}$, where $q$ is the charge of particles and $\mathbf{B}$ denotes the pseudo-vector of the electromagnetic field. We set $\mathbf{B}$ to $[0.5, 0.5, 0.5]$ here. 

\paragraph{Implementation details}
Following EGNN, for each system, we sample 3,000 trajectories for training, 2,000 for validation and 2,000 for test. 
For each training sample, we provide the initial particle positions $\mX(0)=\{\vx_1(0), \cdots \vx_N(0)\}$, velocities $\mV(0)=\{\vv_1(0), \cdots \vv_N(0)\}$ and associated charges $q=\{q_1, \cdots q_N\} \in \{1, -1\}^5$ as the model input. The prediction label is particle positions after 1,000 timesteps. 
We compare our model ClofNet to a canonical non-equivariant graph neural network (GNN) 
and four equivariant models, Radial Field \cite{kohler2019equivariant}, TFN, SE(3)-Transformer and EGNN that can output equivariant vectors. We implemented ClofNet by taking the same message-passing block as EGNN. All baselines consist of 4 layers with hidden dimension 64 and are trained with AdamW optimizer \citep{loshchilov2017decoupled} via a Mean Squared Error (MSE) loss. The learning rate and training epochs are tuned independently for each model.
Note that we add the reduced centroid back to the predicted positions of ClofNet to achieve translation equivariance. Further details about data generation and model implementation are provided in Appendix \ref{app:newton}. An additional experiment of a non-Markov dynamical system can be found in Appendix \ref{partial}.

\paragraph{Results} 
As shown in Table \ref{table:nbody}, our model significantly outperforms all baselines on all datasets, demonstrating the expressiveness of ClofNet on modeling geometrical systems. More importantly, compared to EGNN, the previous state-of-the-art (SOTA) of the task, ClofNet exhibits stronger modeling capacity on two more challenging force field scenarios, illustrating that ClofNet can effectively model complicated force directions besides the radial direction. Compared with TFN and SE(3) Transformer, the remarkable shorter forward time of ClofNet indicates that ClofNet is also more time-efficient than equivariant models based on spherical harmonics. Since it only processes tensors in the original space, not including any complex equivariant functions. 

\paragraph{Analysis for different number of training samples} Following \citep{satorras2021en}, we also conduct extensive experiments to analyze the performance of ClofNet on datasets in the small and large data regime. The results on the G+ES dataset with sample numbers from 200 to 50,000 are summarized in Figure \ref{fig:sample_complexity}. The proposed ClofNet outperforms another equivariant network EGNN in both small and large data regimes, implying the strength of ClofNet on modeling complex force fields. Compared to GNN, ClofNet exhibits a significant strength in the small data regime and achieves comparable performance in the large data regime, demonstrating that ClofNet is more data-efficient than non-equivariant models in geometric scenarios since the SE(3) equivariance is explicitly encoded into our model. Specifically, ClofNet still keeps competitive performance even with $0.4\%$ (200 of 50,000) training samples, showing its superiority on sample complexity.
\begin{figure}
    \centering
    \includegraphics[width=0.92\linewidth]{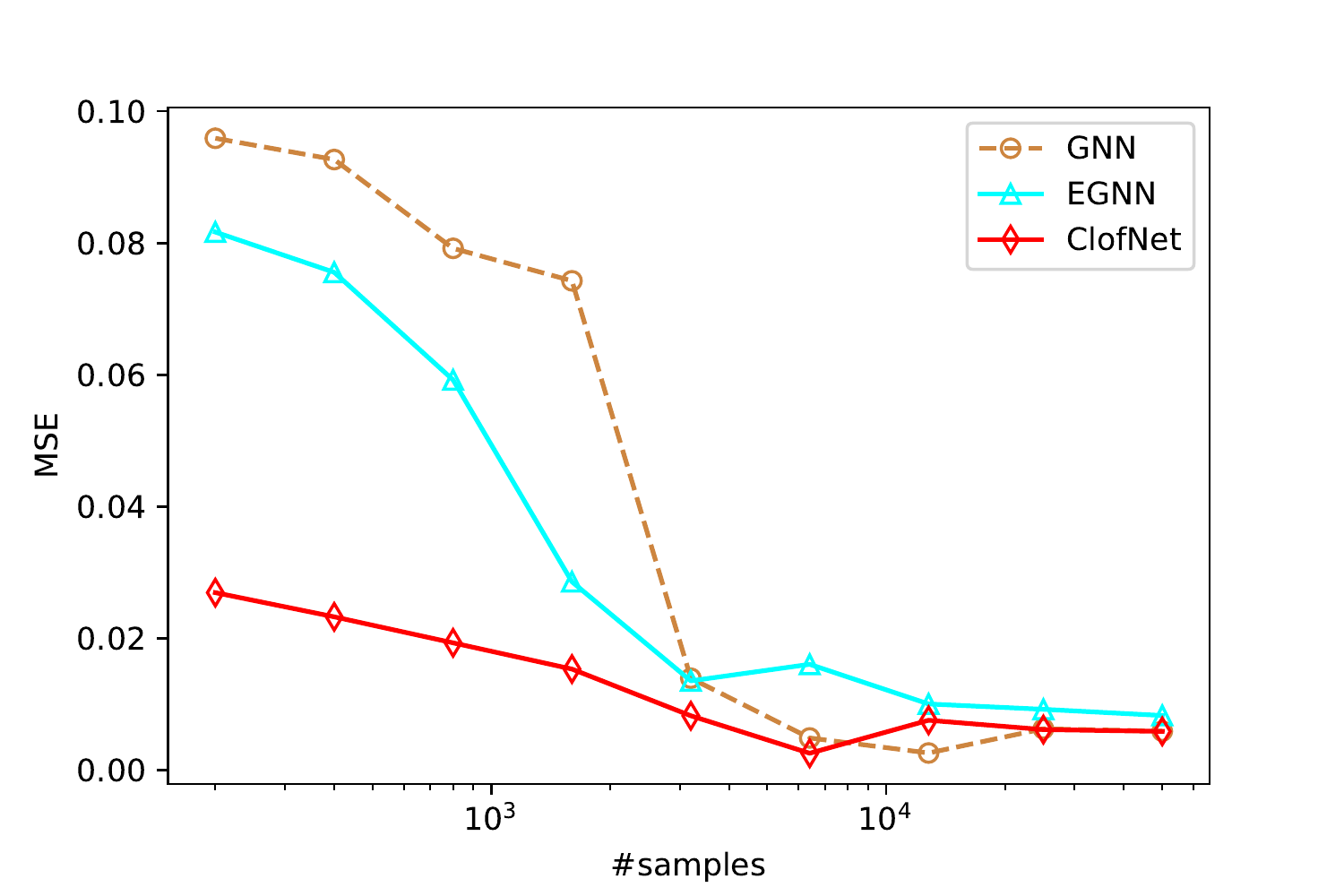}
    \caption{MSE on the G+ES dataset for GNN, EGNN and ClofNet when sweeping over different amounts of training data.}
    \label{fig:sample_complexity}
    \vskip -0.1in
\end{figure}

\subsection{Molecular Systems}
To validate the effectiveness of ClofNet in real-world scenarios, we evaluate ClofNet on a fundamental task in molecular systems called equilibrium conformation generation, trying to predict stable 3D structures from 2D molecular graphs. One of the existing SOTA approaches tackles the problem by directly estimating the gradient fields of the log density of atomic coordinates and then using estimated gradient fields to generate stable structures \citep{shi2021learning, xu2021learning}. The key challenge with this approach is that such gradient fields of atomic coordinates are SE(3) equivariant. 
Following this paradigm, we leverage ClofNet to estimate the gradient fields of atomic coordinates to illustrate its modeling capacity on complex molecular systems.

\paragraph{Evaluation Tasks} We conduct experiments on two tasks: (1) \textit{Conformation Generation} evaluates the capacity of ClofNet to learn the conformation distribution by measuring the diversity and accuracy of generated conformations. Given the RMSD of heavy atoms that measures the distance between generated conformation and the reference, Coverage (COV) and Matching (MAT) scores are defined to measure the diversity and accuracy for a given RMSD threshold $\delta$ respectively.
\begin{equation}
\text{COV}(S_g, S_r) = \frac{1}{|S_r|}|\{R \in S_r|\text{RMSD}(R,\hat{R})<\delta, \hat{R} \in S_g\}|,
\end{equation}
\begin{equation}
\text{MAT}(S_g, S_r) = \frac{1}{|S_r|}\sum_{R \in S_r} \min_{\hat{R} \in S_g} \text{RMSD}(R, \hat{R}),
\end{equation}
where $S_g$ and $S_r$ denote generated and reference conformations, respectively. 

(2) \textit{Distributions Over Distances} evaluates the discrepancy of distance geometry between the generated conformations and the reference conformations. Following \citep{shi2021learning}, we utilize maximum mean discrepancy (MMD) \citep{gretton2012kernel} to measure the discrepancy between the generated distributions and the reference distributions.


\begin{table}[t]
  \caption{COV and MAT scores of different approaches on GEOM-QM9 and GEOM-Drugs datasets. For the COV score, the threshold $\delta$ is set to $0.5 \angstrom$ for QM9 and $1.25 \angstrom$ for Drugs.}
  \label{table:molecular}
  \centering
  \begin{adjustbox}{max width=1\linewidth}
  \begin{tabular}{c|c|cccc}
    \toprule
      \multirow{2}{*}{Dataset} & \multirow{2}{*}{Method} 
      & \multicolumn{2}{c}{COV ($\%$)$\uparrow$} & \multicolumn{2}{c}{MAT ($\angstrom$)$\downarrow$}\\
      & & Mean & Median & Mean & Median\\
    \midrule
    \multirow{6}{*}{QM9} & RDKit & $83.26$ & $90.78$ & $0.3447$ & $0.2935$\\
    & CGCF & $78.05$ & $82.48$& $0.4219$& $0.3900$\\
    & ConfGF & $88.49$ & $94.13$& $0.2673$& $0.2685$\\
    & GeoMol & 71.26 & 72.04 & 0.3730 & 0.3731 \\
    & DGSM & $\textbf{91.49}$ & $\textbf{95.92}$ & $\textbf{0.2139}$ & $\textbf{0.2137}$ \\
    & $\text{EGNN}$ & $80.93$ & $86.27$ & $0.3832$ & $0.3898$ \\
    & ClofNet & $90.21$ & $93.14$ & $0.2430$ & 0.2457 \\
    \midrule
    \multirow{6}{*}{Drugs} & RDKit & $60.91$ & $65.70$ & $1.2026$ & $1.1252$ \\
    & CGCF & $53.96$& $57.06$& $1.2487$& $1.2247$\\
    & ConfGF & $62.15$& $70.93$& $1.1629$& $1.1596$ \\
    & GeoMol & 67.54& 72.71& 1.0325& 1.0580 \\
    & DGSM & 78.73 & 94.39 & 1.0154& 0.9980 \\
    & $\text{EGNN}$ & $40.71$ & $33.01$ & $1.3574$ & $1.3346$ \\
    & ClofNet & $\textbf{88.64}$ & $\textbf{97.56}$ & $\textbf{0.9040}$ & $\textbf{0.9023}$\\
    \bottomrule
  \end{tabular}
  \end{adjustbox}
  \vskip -0.1in
\end{table}

\begin{table}[t]
    \caption{Accuracy of the distributions over distances generated by
different approaches compared to the ground-truth.}
    \label{table:ensemble_prop}
    \centering
    \begin{adjustbox}{max width=1\linewidth}
    \begin{tabular}{c|cccccc}
    \toprule
    \multirow{2}{*}{Method} & \multicolumn{2}{c}{Single} & \multicolumn{2}{c}{Pair} & \multicolumn{2}{c}{All} \\ 
    & Mean & Median & Mean & Median & Mean & Median\\\midrule
    RDKit & 3.4513 & 3.1602 & 3.8452 & 3.6287 & 4.0866 & 3.7519\\
    CGCF & 0.4490 & 0.1786 & 0.5509 & 0.2734 & 0.8703 & 0.4447\\
    ConfGF & 0.3684 & 0.2358 & 0.4582 & 0.3206 & 0.6091 & 0.4240\\\midrule
    ClofNet & $\textbf{0.1317}$ &\textbf{0.0420}&\textbf{0.1787}&\textbf{0.0695}&\textbf{0.3185}& \textbf{0.1142}\\
    \bottomrule
    \end{tabular}
    \end{adjustbox}
\end{table}

\begin{table}[t]
  \caption{Ablations for ClofNet on two datasets.}
  \label{table:ablation}
  \centering
  \begin{adjustbox}{max width=1\linewidth}
  \begin{tabular}{c|c|cccc}
    \toprule
      \multirow{2}{*}{Dataset} & \multirow{2}{*}{Method} 
      & \multicolumn{2}{c}{COV ($\%$)$\uparrow$} & \multicolumn{2}{c}{MAT ($\angstrom$)$\downarrow$}\\
      & & Mean & Median & Mean & Median\\
    \midrule
    \multirow{4}{*}{QM9} &  CN $w/o$ Sca & $88.99$ & $\textbf{94.55}$ & $0.3050$& $0.3066$\\
    & CN $w/o$ GT & $85.21$ & $91.18$ & $0.3057$ & $0.3060$\\
    & ClofNet & $\textbf{90.21}$ & $93.14$ & $\textbf{0.2430}$ & $\textbf{0.2457}$\\
    \midrule
    \multirow{4}{*}{Drugs} &  CN $w/o$ Sca & $65.67$& $75.63$ & $1.1410$ & $1.1132$\\
    & CN $w/o$ GT & $70.57$ & $81.82$& $1.1075$& $1.1004$\\
    & ClofNet & $\textbf{88.64}$ & $\textbf{97.56}$ & $\textbf{0.9040}$ & $\textbf{0.9023}$ \\
    \bottomrule
  \end{tabular}
  \end{adjustbox}
  \vskip -0.1in
\end{table}

\paragraph{Datasets} Following \citep{xu2021learning, shi2021learning}, we evaluate the proposed model on the GEOM-QM9 and GEOM-Drugs datasets \citep{axelrod2020geom} as well as the ISO17 dataset \citep{schutt2017schnet}. The QM9 dataset consists of small molecules up to 29 atoms, while the molecules in the Drugs dataset are generally of medium size, containing up to 181 atoms. ISO17 is for the distance modeling task, which consists of conformations of various molecules with the atomic composition
$\text{C}_7\text{H}_{10}\text{O}_2$.
To keep a fair comparison with the existing SOTA method ConfGF \citep{shi2021learning}, we reproduce its data collection and split settings rigorously. Further details are described in Appendix \ref{app:mcg}. 

\paragraph{Implementation details} For the training procedure, we adopt the same score-based framework as in ConfGF. In detail, we perturb the atomic positions with
random Gaussian noise of various magnitudes and estimate the gradient field by training the noise conditional score network that is constructed based on ClofNet. 
See more details about score-based networks in \citep{shi2021learning,song2019generative} or Appendix \ref{app:mcg}.  The ClofNet is equipped with 4 Graph Transformer blocks and the hidden dimensions are set to 288. All models are trained with Adam optimizer via the score matching loss function (See Appendix \ref{app:mcg}, \ref{score}) for 400 epochs. For each molecule in the test set, we follow ConfGF to sample twice as many conformations as the reference ones from each model. All hyperparameters of the score-based framework are provided in Appendix \ref{app:mcg}.

We compare ClofNet to six classic methods for conformation generation. Specifically, both RDKit \citep{landrum2013rdkit} and CGCF\citep{xu2021learning} are distance-based approaches. ConfGF is most close to us, attempting to generate conformations by learning the gradient field of the data distribution in an equivariant manner. DGSM \citep{luo2021predicting} is another gradient field-based method built upon ConfGF, which proposes a dedicated dynamic graph constructing strategy to model long-range interactions. However, both ConfGF and DGSM only utilize the distance matrix as the geometric input. GeoMol \citep{ganea2021geomol} is a two-stage conformation generation strategy. We also reproduce EGNN on this task as our baseline. The empirical results of all baselines except GeoMol are copied from the original papers since they use the same data split setting as us. We locally reproduce the results of GeoMol on our data split.

\paragraph{Results} For the \textit{Conformation Generation} task,  we summarize the mean and median COV and MAT scores on two benchmarks for all methods. As shown in Table \ref{table:molecular}, ClofNet achieves the best or the second best performance on almost all metrics and datasets, demonstrating the effectiveness of our proposed method. In particular, on the Drugs dataset,
compared with ConfGF which employs a similar learning strategy with us, ClofNet significantly increases $26.5\%$ COV-mean and $26.6\%$ COV-median scores. Notably, ClofNet also achieves better performance than DGSM which leverages a more complex dynamic graph constructing strategy. A potential interpretation is that molecules in Drugs usually contain more atoms and complex chemical functional groups (e.g., Benzene rings) than those of QM9, thus distance-based geometry is not sufficient to model the gradient field of this complex distribution.

ClofNet also achieves significant improvement in the \textit{Distributions Over Distances} task. As shown in Table \ref{table:ensemble_prop}, ClofNet dramatically outperforms the previous SOTA (ConfGF), demonstrating the strong capacity of the proposed model in modeling molecular dynamics data. In particular, ClofNet reduces the MMD by a magnitude in both Single-median and Pair-Median metrics.

\paragraph{Ablations}
Although the superior performance on multiple tasks verifies the effectiveness of ClofNet, it remains unclear whether
the proposed strategies make a critical contribution. 
In light of this, we set up several ablative configurations and list the empirical results in Table \ref{table:ablation}.
For the scalarization block, we conduct a variant of ClofNet without the scalarization block, named CN\textit{ $w/o$} Sca, which only takes the distance matrix of molecules as the input geometric feature. The results show that including scalarization block plays an important role in the model, as the COV-mean and COV-median scores on the Drugs dataset increase by $23.0\%$ and $21.9\%$, respectively, which implies that including all geometric information will boost the performance of the model.
For the graph transformer block, we replace the block with the GIN network \citep{xu2018powerful} that is employed in ConfGF, getting the variant CN\textit{ $w/o$} GT. The results indicate that introducing the attention mechanism will also contribute to the gradient field modeling. We cannot conduct the ablative study for the vectorization block because it guarantees the output of ClofNet is an equivariant vector field.

\section{Conclusion and Future Work}
To express geometric information more efficiently and flexibly,
we develop a novel SE(3) equivariant neural network with complete local frames (ClofNet), aiming at lossless utilization of tensors without incorporating high-dimensional equivariant function embedding. 
With the proposed local frames, ClofNet could cooperate with any graph neural networks without concerns about breaking the equivariance symmetry. Theoretical analyses and extensive empirical results verify the effectiveness of ClofNet. In the future, we will extend our strategy to other (local) symmetry groups.

\section*{Acknowledgements} We thank all the anonymous reviewers for their valuable comments. 
The work was supported partly by the
National Natural Science Foundation of China under Grant 62088102.

\bibliography{reference}
\bibliographystyle{icml2022}

\newpage
\appendix
\newpage
\appendix
\onecolumn
\section{Appendix}
\subsection{The special Eulidean group SE(3)} \label{so3}
The special Euclidean $SE(3)-$group (\citep{marsden2013introduction}
is defined as a semidirect (noncommutative) product of 3D rotations $SO(3)$ and 3D
translations, $SE(3):=SO(3)\rhd \mathbb{R}^{3}$.  An element of $SE(3)$ is a pair $(R,\va)$ where $R\in SO(3)$ and $\va\in \mathbb{R}^{3}.$ 

The action of $SE(3)$ on $\mathbb{R}^{3}$ is the rotation $R$
followed by translation by the vector $a$ and has the expression
\begin{equation*}
(R,\va)\cdot x=Rx+\va.
\end{equation*}
Then the multiplication and inversion operation in $SE(3)$
are given by
\begin{equation*}
(R_1,\va)(R_2,\vb)=(R_1R_2,R_1 \vb+\va)\qquad \text{and\qquad }(R_1,\va)^{-1}=(R_1^{-1},-R_1^{-1}\va),
\end{equation*}
for $R_1,R_2\in SO(3)$ and $\va,\vb\in \mathbb{R}^{3}.$ The identity element is $(I,0)$.

The 3D rotation group $SO(3)$ is defined by $$SO(3)=\{R\in \mathcal{M}_{3\times 3}(\mathbb{R}):R^{T}R=I,\det
(A)=1\}.$$ For example, 
\begin{equation*}
R_{\varphi } =\left[
\begin{array}{ccc}
1 & 0 & 0 \\
0 & \cos \varphi & -\sin \varphi \\
0 & \sin \varphi & \cos \varphi%
\end{array}
\right] 
\end{equation*}
belongs to $SO(3)$, parameterized by $\varphi \in [0,2\pi]$.

\subsection{Torsion force field} \label{torsion}
In this section, we consider a 4-body system $X = (\vx_l,\vx_i,\vx_j,\vx_k)$.
\subsubsection{Dihedral angle} \label{angle}
For a given node $\vx_l$ with three neighbors $\vx_i$, $\vx_j$ and $\vx_k$. The dihedral angle $\theta$ of the plane spanned by $(\vx_l - \vx_i,\vx_l - \vx_j)$ and the plane spanned by $(\vx_l - \vx_i, \vx_l - \vx_k)$ is given by the inner product of normal vectors of the two planes: $$\frac{(\vx_l -\vx_i) \times (\vx_l -\vx_j)}{\norm{(\vx_l -\vx_i) \times (\vx_l -\vx_j)}},$$ and $$\frac{(\vx_l -\vx_i) \times (\vx_l -\vx_k)}{\norm{(\vx_l -\vx_i) \times (\vx_l -\vx_k)}}.$$
That is,
\begin{equation} \label{cos}
\cos \theta: = \frac{(\vx_l -\vx_i) \times (\vx_l -\vx_j) \cdot (\vx_l -\vx_i) \times (\vx_l -\vx_k)}{\norm{(\vx_l -\vx_i) \times (\vx_l -\vx_j)} \norm{(\vx_l -\vx_i) \times (\vx_l -\vx_k)}}.
\end{equation}
Therefore, the dihedral angle is a function of the position vectors $\vx_l$, $\vx_i$, $\vx_j$ and $\vx_k$, which cannot be determined 
purely by the norm of positions: $\norm{\vx_l - \vx_i}$,$\norm{\vx_l - \vx_j}$ and $\norm{\vx_l - \vx_k}$.
\subsubsection{Torsion angle and torsion force}
The torsion angle of $x_j - x_i$ and $x_l -x_k$ with respect to the bond $x_i - x_k$ is just the  dihedral angle $\theta$ between the plane spanned by $(\vx_l - \vx_i,\vx_l - \vx_j)$ and the plane spanned by $(\vx_l - \vx_i, \vx_l - \vx_k)$. By Fourier series, the torsional energy has the following form:
$$E_{tor}(\theta):= \sum_n V_n \cos{n \theta}.$$
By cutting off the high frequency part ($n \leq 2$), we can simulate an effective torsional energy. Then the torsion force of a particle $\vx$ is given by the gradient vector:
$$F(\vx) = - \frac{\partial E_{tor}(\theta)}{\partial \vx},$$
where $\cos{\theta}$ is defined by (\ref{cos}). 
\subsection{Additional descriptions on blocks of ClofNet}
\subsubsection{Graph message-passing Block}
\label{app:ga_b}

In this section, we provide the implementation details of the Graph Message-passing Block (GMP) that is utilized in ClofNet.
For a many-body system $\mX$, we first represent it as a spatial graph and utilize a message-passing based network to learn the SO(3)-invariant edgewise embeddings $m_{ij}$ from the graph. Considering that there does not exist any graph topology in most real-world scenarios, we design a Graph Transformer Block (GTB) that leverages the attention mechanism due to its powerful capacity in learning the correlations between inter-instances \citep{vaswani2017attention}. Now we discuss the workflow of the naive message-passing block and the GTB block. Note that here we omit the time index of the SO(3)-invariant scalars for the simplicity of description.
\paragraph{Feature embedding}
Given geometric scalars $s_{ij}$, node features $h_i$ and edge features $m_{ij}$, GMP first embeds them into high-dimensional representations: 
\begin{gather}
    h_i = \mathrm{MLP}(h_i),\\
    e_{ij} = \mathrm{MLP}(e_{ij}), \\
    s_{ij} = \mathrm{Fourier}(s_{ij}) \text{ or } \mathrm{MLP}(s_{ij}),\\
\end{gather}
where $\mathrm{MLP}$ denotes a fully connected network and $\mathrm{Fourier}$ denotes a Fourier transformation with a tuple of learnable frequencies. 

\paragraph{Naive message-passing Block} \label{naive} Following \citep{satorras2021en}, each message-passing block takes as input the set of node positions $\vx^l$, node embeddings $h^l$, SO(3)-invariant scalars $s_{ij}$ and edge information $e_{ij}$ and outputs a transformation on $h^{l+1}$ and $x^{l+1}$. The equations for each message-passing layer are defined as follows:
\begin{gather}
    m_{ij} = \phi_m (s_{ij}, h_i^l, h_j^l, e_{ij}), \\
    h_i^{l+1} = \phi_h  (h_i^{l},\sum_{j \in \mathcal{N}(i)} m_{ij}), \\
    \mathcal{F}_{ij}^{l+1}=\mathbf{EquiFrame}(\vx_i^l, \vx_j^l), \\
     \vx_i^{l+1} = \vx_i^{l} + \frac{1}{N}\sum_{j \in \mathcal{N}(\vx_i)}{\mathbf{Vectorize}(m_{ij}^l, \mathcal{F}_{ij}^l)},
\end{gather}
where $\phi_m$, $\phi_k$ and $\phi_h$ are $\mathrm{MLP}$s with different parameters. Comparing with (\ref{eq:egnnmessages}), the naive message-passing algorithm can also be seen as EGNN \citep{satorras2021en} equipped with our equivariant local frames.

\paragraph{Graph Transformer Block} \label{attention}
The overall architecture of our GTB block is inspired by \citep{shi2020masked}. For each GTB block, we first compute the edgewise message with $\phi_m$ (i.e., Eq. (17), then leverage the message embeddings and node embeddings with a transformer encoder-like architecture to refine the node embeddings. After that, we update edgewise messages with a residue block. In practice, we replace Eq. (18) with the following operations:
\begin{gather}
    q_{i} = \phi_q (h_i^l), k_{ij} = \phi_k (h_i^l, m_{ij}), v_{ij} = \phi_v (m_{ij}), \\
    \alpha_{ij} = \frac{\left<q_{i}, k_{ij}\right>}{\sum_{j' \in \mathcal{N}(i)}{\left<q_{i}, k_{ij'}\right>}}, \\
    \mathcal{M}_i = \mathrm{LayerNorm}(\sum_{j \in \mathcal{N}(i)}{\alpha_{ij}v_{ij}}), \\
    h_i^{l+1} = \phi_h (h_i^{l}, \mathcal{M}_i), \\
    h_i^{l+1} = h_i^{l} + \mathrm{LayerNorm}(h_i^{l+1}),
\end{gather}
where $\phi_q$, $\phi_k$ and $\phi_v$ are linear projection layers. $\phi_h$ is an $\mathrm{MLP}$. $\alpha_{ij}$ and $\mathcal{M}_i$ denote the attention weights and the refined nodewise embeddings, respectively. $\mathrm{LayerNorm}$ refers to the normalization layer adopted in \citep{vaswani2017attention}.

\subsubsection{Learning gradient fields as score-based models} \label{evo}
Molecular conformation generation adopts an SE(3)-equivariant vector field(i.e., gradient field) to depict its dynamic evolving mechanisms.
\paragraph{Gradient field} Gradient field is a widely used terminology meaning the first-order derivative w.r.t. a scalar function \citep{jost2008riemannian,song2019generative}. To generate molecular conformations (i.e. equilibrium state) with a single stage, \citep{shi2021learning} define ``gradient field'' to serve as pseudo-forces acting on each particle. By evolving the particles following the direction of the gradient field, the non-equilibrium system will finally converge to an equilibrium state. 

\paragraph{Equivariant score function and the Evolving block}
As to the statistical ensemble system, we try to predict the reverse evolving process from a random state to equilibrium by integrating a first-order equivariant vector field. For example, all physical-allowable molecule conformations are located in an equilibrium state determined by the energy function. Suppose the forward process  from equilibrium to non-equilibrium of the system satisfies:
\begin{equation} \label{forward}
d\mX(t) = f(\mX(t),t)dt + g(t)dW_t,  \ \ 0 \leq t \leq T,\end{equation}
where $W_t$ is the Brownian motion and the initial state $\mX(0)$ follows an unknown equilibrium distribution $p_0$. Denote the marginal distribution at time $t$ by $p_t$, then we can write $p_t$ in Gibbs distribution form:
$$p_t(\mX(t)) = \exp\{-\beta \mH_t(\mX(t))\}.$$
It implies that the Hamiltonian function $\mH_t$ at time $t$ is entangled with $p_0$.
 According to \cite{song2020score}), the reverse evolving process satisfies the following Langevin dynamics:
\begin{equation} \label{langevin}
d\mX(T-t) = f(\mX(T-t),T-t)dt - g^2(T-t)\nabla \mH_{T-t}(\mX_{T-t})dt + g(T-t)dB_t,\end{equation}
where $B_t$ denotes the standard Brownian motion.
The gradient field of the Hamiltonian function $\nabla \mH_t(x)$ is also called the \textbf{force field}, or the score function \citep{song2019generative}. Therefore, ClofNet $\phi$ is implemented to model the score function:
: $$\nabla \mH_{T-t}(\mX_{t}) = \phi(\mX(t)),\ \ 0 \leq t \leq T.$$
Note that the vector-valued function $f(x,t)$ and the scalar function $g(t)$ is set to be fixed as prior knowledge. 
In \cite{shi2021learning} and our molecular experiment, we use the discretization of VE SDE \citep{song2019generative,song2020score}, where $f\equiv 0$ and
$$g(t) = \sqrt{\frac{[d\sigma^2(t)]}{dt}}.$$
Since the distribution Brownian motion is SO(3)-invariant and $g(t)$ is chosen to be an SO(3)-invariant scalar function, then by (\ref{forward}), the forward process is SO(3)-invariant. Combining with the fact that the initial distribution $p_0$ (molecular conformation distribution) is SO(3)-invariant, it's easy to derive that the score function is a gradient of a \textbf{SO(3)-invariant} function, which means that the score function is \textbf{SO(3)-equivariant}. 

Finally, the learned score function is connected to the evolving block for numerical integrating the Langevin diffusion (\ref{langevin}).
\textbf{For a fair comparison}, ClofNet and \citet{shi2021learning} use the same annealed Langevin dynamics sampling algorithm \citep{song2019generative,shi2021learning}. 

\subsubsection{Scalability} \label{SCALABILITY}
Compared to other efficient equivariant models like EGNN \citep{satorras2021en,satorras2021nn}, the extra computational cost of ClofNet is to calculate the values of the scalars in $t_{ij}, \forall i,j$ through the scalarization Block. Let $E$ denotes the number
of edges in the graph, then the cost of calculating the scalars (along with the frames) is $\mathcal{O}(3 \times3 \times E)$, which is of the same order as the cost of performing 1-hop message-passing. Therefore, the extra cost is much less, Compared to the
computational cost of back-propagation in the neural networks. For example, for the molecular conformation generation task, transforming the
tensors into SO(3)-invariant scalars only brings 9.6$\%$ extra real-time computational cost and 17.4$\%$ extra memory cost.

\subsection{Proof and discussions} 
\label{app:proof}
\begin{proposition}
\label{exp}
The complete frame $(\va^t,\vb^t,\vc^t)$ defined by (\ref{basis}) is equivariant under SO(3) transformation of the spatial space.    
\end{proposition}
\begin{proof}
Let $g \in SO(3)$, then under the action of $g$, the positions of the many-body system $X(t)$ transform equivariantly:
$$(\vx_1(t),\dots,\vx_n(t)) \xrightarrow{g} (g\vx_1(t),\dots,g\vx_n(t)).$$
From the definition of $\va^t$, we know that
$$\va^t \xrightarrow{g} g\va^t.$$ 
For $\vb^t$, since
\begin{align} \label{cr}
(g\vx_i(t))\times (g\vx_j(t)) & = \det (g) (g^T)^{-1} ( \vx_i(t) \times \vx_j(t)) \\
& = g(\vx_i(t) \times \vx_j(t)),
\end{align}
where we have used $g^{-1} = g^T$ for orthogonal matrix $g$ to get the last line. Therefore, $\vb^t \xrightarrow{g} g\vb^t.$ Applying (\ref{cr}) once again, we have  $\vc^t \xrightarrow{g} g\vc^t.$
\end{proof}
\paragraph{Frame bundle and scalarization technique}

A frame at $x \in M$ is a linear isomorphism $u $ from $\mathbb{R}^d$ to the tangent space at $x$:$T_x M$. We use $F(M)_x$ to denote the space of all frames at $x$. Then $GL(d,\mathbb{R})$ acts on 
$F(M)_x$ by
$$\mathbb{R}^d \xrightarrow{g} \mathbb{R}^d \xrightarrow{u} T_x M.$$
Then the frame bundle
$$F(M) = \cup_{x \in M} F(M)_x$$
can be made into a differential manifold. From the principle bundle point of view, each differential manifold $M$ is a quotient of its frame bundle $F(M)$ by the general linear group $GL(d,\mathbb{R})$: $M=F(M)/GL(d,\mathbb{R})$.
We denote the quotient map by $\pi$: 
$F(M) \xrightarrow{\pi}  M.$
Then, each point of $\vu \in F(M)$ is a reference frame located at $\vx: = \pi(u) \in M$. On the other hand, $\mathbb{R}^d$ can be seen as a $d$-dimensional differential manifold with a global coordinates chart and  SO(3) is the structure-preserving group of the Euclidean metric with a fixed orientation.  

Following \citet{hsu2002stochastic}, let $\{e_i,\ 1 \leq i \leq d\}$ be the canonical frame of $\mathbb{R}^d$, and $\{e^i\}$ the corresponding dual frame. At each frame $u$, the vectors $Y_i : = u e_i$ form a frame of $T_x M$. Let $\{Y^i\}$ be the dual frame of $T_x^* M$, then a (r,s)-tensor $\theta$ can be expressed as
\begin{equation} \label{tensor field}
\theta = \theta^{i_1 \cdots i_r}_{j_1 \cdots j_s} Y_{i_1}\otimes \cdots \otimes Y_{i_r} \otimes Y^{j_1}\otimes \cdots \otimes Y^{j_s}.\end{equation}
The \textbf{scalarization} of $\theta$ at $u$ is
\begin{equation}
\label{scal}
\Tilde{\theta}(u) : = \theta^{i_1 \cdots i_r}_{j_1 \cdots j_s} e_{i_1}\otimes \cdots \otimes e_{i_r}\otimes e^{j_1}\otimes \cdots \otimes e^{j_s}.\end{equation}
Through scalarization, a tensor field $\theta$ becomes an ordinary vector space valued function on $F(M)$:
$$\Tilde{\theta} : F(M) \rightarrow \mathbb{R}^r \otimes \mathbb{R}^s.$$
Therefore, geometric operations such as covariant derivative and tensor product on manifolds can be realized as directional derivative and tensor product of ordinary vector spaces \citep{hsu2002stochastic}. 
\begin{proposition} \label{corr}
There is a one-to-one correspondence between the scalarization of tensor fields on the orthonormal frame bundle $O(\mathbb{R}^3)$ (\ref{scal}) and the SO(3)-invariant scalars tuple obtained by the scalarization block (\ref{scalar}) under an equivariant frame.
\end{proposition}
\begin{proof}
Since we are working in $\mathbb{R}^3$ with a fixed orientation, the $GL(3,\mathbb{R})$ group action is reduced to $SO(3)$ global group action acting on the orthonormal frame bundle $O(\mathbb{R}^3)$. A frame $u \in O(\mathbb{R}^3)$ at $\pi (u) \in \mathbb{R}^3$ can be transported to another point in $\mathbb{R}^3$ by translation. Therefore, we neglect the origin of the frame in the proof. Let $u_e =(\ve_1,\ve_2,\ve_3)$ be an equivariant frame of $\mathbb{R}^3$, then given a scalars tuple $\{\theta^{i_1,\dots,i_r}\}$, we construct a vector-valued function on $O(\mathbb{R}^3)$ by:
$$\Tilde{\theta}^{i_1,\dots,i_r}(u) = g_{i_1i_1'}\cdots g_{i_ri_r'}\theta^{i_1',\dots,i_r'},$$
where $g$ is the $SO(3)$ transformation from $u_e$ to another frame $u$. Moreover, $\Tilde{\theta}^{i_1,\dots,i_r}(u)$ is SO(3)-equivariant, since 
$$\Tilde{\theta}(gu) = g\Tilde{\theta}(u),$$
where the $g$ on the right side means the usual extension of the action of SO(3) from $\mathbb{R}$ to the tensor space $\mathbb{R}^{\otimes r}$. We have constructed the $(r,\,0)$ tensor field from the scalars tuple. It's easy to check that $\Tilde{\theta}^{i_1,\dots,i_r}(u)$ induces a $(r,\,0)$ tensor field on $\mathbb{R}^3$ by following the definition of (\ref{ten}). 

From scalarization $\Tilde{\theta}^{i_1,\dots,i_r}(u)$ on $O(\mathbb{R}^3)$ to scalarization is obvious. Note that any equivariant frame $u_e$ is also a point on $O(\mathbb{R}^3)$, therefore the scalars tuple is just the values of $\Tilde{\theta}^{i_1,\dots,i_r}$ at $u_e$:
$$\theta^{i_1,\dots,i_r} = \Tilde{\theta}^{i_1,\dots,i_r}(u_e).$$
For general $(r,\,s)$-type tensors, the proof is the same by adding the dual frame of $u_e$.
\end{proof}
\subsubsection{Permutation and translation equivariance} \label{symmetry}
For a many-body system $\mX(t) = (\vx_1(t),\dots,\vx_n(t))$, the centroid $C(t)$ is defined by
$$C(t) = \frac{\vx_1(t) + \cdots \vx_n(t)}{n}.$$
Translating the reference by $\vh$, then
$$X(t) + h \rightarrow C(t) + h.$$
Therefore, recentering the reference's origin to the centroid at the input's time $t=0$, we have
$$\mX(t) - C(0) \xrightarrow[]{\text{translation by }\ \vh\ \text{at}\ t=0} \mX(t) - C(0).$$
That is, the system is translation-invariant under the re-centered reference if the translation is done at the input's time $t=0$, which is exactly the scenario considered in predicting the future trajectory or state. Note that although most tensors are \textbf{translation-invariant} (e.g., mass,velocities, Gravitational acceleration, gradient fields), the future positions of the many-body system should be \textbf{translation-equivariant}. Therefore, we add $\vc(0)$ back to the output of the predicted positions. 

Permutation equivariance is automatically guaranteed by all \textbf{graph-based} algorithms, because isomorphic graphs are obliged to yield isomorphic outputs. Therefore, we will just briefly mention why the common 1-hop message-passing algorithms are Permutation equivariant. To emphasis the permutation symmetry, we quote a conclusion from \citep{zaheer2017deep} (which is derived from the Kolmogorov–Arnold representation theorem): if $f(\vx_1,\dots,\vx_n)$ is a permutation invariant multivariate continuous function, then 
\begin{equation} \label{agg}
    f(\vx_1,\dots,\vx_n) = g(\sum_{i=1}^n \phi(\vx_i)).
\end{equation}

The crucial point is the function $\phi$ is shared among all the points, which exactly fits the definition of the so-called message-passing scheme.
For a many-body system $X$, let $v=(v_1, \dots v_n)$ be its vector field, then $v_i \in \mathbb{R}^3$ corresponds the equivariant vector for $\vx_i \in \{\vx_1,\dots,\vx_n\}$. Denote ClofNet with parameters $\theta$ by $\Phi^{\theta,t}=\{\Phi_i^{\theta}\}_{i=1}^n$, then 
$$v_i(t) = \Phi_i^{\theta}(\mX(t),t),$$
for a fixed particle $\vx_i$. 
Suppose $(i_1,\dots,i_k)$ (neighbors of $\vx_i$) indicate indexes of particles which have interaction with $\vx_i$, then obviously $1 \leq k \leq n-1$. By (\ref{agg}), $\Phi_i(\mX(t),t)$ is an aggregation of message from $\vx_i$'s neighbors, therefore we have:
\begin{equation} \label{vag}
\Phi_i(X(t),t) = \frac{1}{k}\sum_{j=1}^k \phi(\vx_i(t), \vx_{i_j}(t),t),\end{equation}
and $\phi$ is a SO(3)-equivariant network with vector output. Note that (\ref{vag}) performs aggregation at the level of vectors, therefore we choose $g$ in (\ref{agg}) to be the arithmetic mean to preserve SO(3) symmetry. A shared neural network $\phi$ guarantees
the permutation equivariance of ClofNet. 
\subsubsection{Equivariance of ClofNet} \label{equi}
We need the following lemma:
\begin{lemma} \label{lem}
Suppose $h$ is an invariant function, $f_1, \dots, f_k$ are arbitrary nonlinear functions. Then, the composition $f_k \circ \cdots \circ f_1 \circ h$ is an invariant function. 
\end{lemma}
\begin{proof}
Consider the group action $g$ acting on $x \in \mathbb{R}^3$, then $h$ is invariant means that $h \circ g (x) := h(gx) = h(x)$. We have
$$f_k  \circ \cdots \circ f_1 \circ h \circ g = f_k  \circ \cdots \circ f_1 \circ h,$$
for every group action $g$. Hence the composition $f_k \circ \cdots \circ f_1 \circ h $ is invariant.
\end{proof}
\paragraph{Proof of proposition \ref{Equivariance}}
Let $f_1$ be the scalarization block. Because the output of $f_1$ are scalars (from (3.3)) and SO(3)-scalars must be invariant under SO(3) group action, we can conclude that $f_1$ is invariant. Denote each graph-transformer layer by $f_i,\ i \in \{2,\dots,k \}$, then by Lemma \ref{lem}, we can conclude that the output of the graph-transformer block is also invariant. \\
Finally, from the standard fact that scalars multiply vectors yield equivariant vectors, we conclude that the output of the vectorization block is SO(3)-equivariant.

Combining with the translation equivariance in \ref{symmetry}, our model is SE(3)-equivariant. 

\begin{remark} \label{re}
In this remark, we investigate how the complete local frame (\ref{basis}) transforms under central reflection $R: \vx \rightarrow - \vx$. First, we can derive $\va \rightarrow - \va$ from the reflection operation. Second, by the right-hand-thread rule, the cross product of two equivariant vectors yields a pseudo-vector:
$\vb = \vx_i \times \vx_j \rightarrow \vb.$
Then it is easy to imply that $\vc \rightarrow - \vc$. Due to $\det{(-\va,\vb,-\vc)} = 1$, the orientation of our local frame under reflection $R$ remains unchanged. Thant means, ClofNet is central reflection anti-symmetric. 
Considering that other reflections including mirror reflection alone a plane can be generated by rotation and central reflection, we can derive that ClofNet is reflection anti-symmetric. 

\end{remark}
\subsection{Information loss} \label{Information loss}
Denote the prediction target by $Y$, then by the information bottleneck principle (equation (3) of \cite{tishby2015deep}), a learning problem can be formulated as finding a minimal sufficient statistics of the input $X$ with respect to $Y$. In mathematical terminology, we have a Markov chain:
$$Y\rightarrow X \rightarrow \hat{X}.$$
$\hat{X}$ is a filtration of the original output $X$ according to a fixed target $Y$, therefore an information lossless model should satisfy the following equality:
\begin{equation} \label{bottle}
I(Y;X) = I(Y,\hat{X}),
\end{equation}
where $I(\cdot | \cdot)$ denotes the mutual information between two random variables. By the data processing inequality, (\ref{bottle}) is satisfied if and only if $I(X,Y|\hat{X}) = 0$ and $Y\rightarrow \hat{X} \rightarrow X$ also forms a Markov chain. 

Consider a toy 4-body example, a particle denoted by $x_l$ with three neighbor particles $x_i$, $x_j$ and $x_k$. Suppose the force undertook by $x_l$ is a function of $\theta$: $F_l(\theta)$, where $\theta$ is the dihedral angle  between the plane spanned by $(\vx_l - \vx_i,\vx_l - \vx_j)$ and the plane spanned by $(\vx_l - \vx_i, \vx_l - \vx_k)$. From the definition, $\cos \theta$ is given by the inner product of normal vectors of the two planes: \begin{equation} \label{1}
\vn_1 : = \frac{(\vx_l -\vx_i) \times (\vx_l -\vx_j)}{\norm{(\vx_l -\vx_i) \times (\vx_l -\vx_j)}},\end{equation} and \begin{equation} \label{2}\vn_2 : = \frac{(\vx_l -\vx_i) \times (\vx_l -\vx_k)}{\norm{(\vx_l -\vx_i) \times (\vx_l -\vx_k)}}.\end{equation}
Therefore, the dihedral angle is a function of the position vectors $\vx_l$, $\vx_i$, $\vx_j$ and $\vx_k$, which cannot be determined 
purely by the positions' norm: $\norm{\vx_l - \vx_i}$,$\norm{\vx_l - \vx_j}$ and $\norm{\vx_l - \vx_k}$.

Suppose we know the coordinates of particles: $X: = (x_l,x_i,x_j,x_k)$, and we want to predict $Y = F_l(\theta)$. Then from the definition of dihedral angle, obviously $Y = F_l(\vx_l - \vx_i, \vx_l - \vx_j, \vx_l - \vx_k)$. Now a natural question to propose is how to decide $\hat{X}$, which satisfies (\ref{bottle}). In EGNN, $\hat{X}_{\text{EGNN}} = (\norm{\vx_l - \vx_i},\norm{\vx_l - \vx_j},\norm{\vx_l - \vx_k})$, then $Y$ and $X$ are not independent conditioned on $\hat{X}$, since $\theta$ is not a function of $\hat{X}$, but a function of $X$. That is,
$$I(Y,X | \hat{X}) \neq 0 \ \ \text{and}\ \ I(Y,X) > I(Y, \hat{X}).$$
Therefore, certain information is lost when transforming the original input $X$  to the input of EGNN:$\hat{X}_{\text{EGNN}}$.

On the other hand, ClofNet transforms $X= (\vx_l - \vx_i,\vx_l - \vx_j,\vx_l - \vx_k)$ to $\hat{X}_{\text{ClofNet}} : = f(X) := (\bar{\vx_l} - \bar{\vx_i},\bar{\vx_l} - \bar{\vx_j},\bar{\vx_l} - \bar{\vx_k})$, where $\bar{\vx_i}$ is the coordinate of $x_i$ under the equivariant frame of ClofNet ( $f$ corresponds to the sclarization, and $\bar{\vx_i} \in \mathcal{S}(X)$ defined in the scalarization section). Since this transformation $f$ is \textbf{invertible} (see (\ref{invert})),
it's obvious that
$$Y\rightarrow \hat{X} \rightarrow X \equiv f^{-1}(\hat{X})$$
forms a Markov chain, and
\begin{align*}
  I(X,Y|\hat{X}) &=  I(f^{-1}(\hat{X}),Y |\hat{X})=0.
\end{align*}
Here we use the fact that $f^{-1}(\hat{X})$ conditioned on $\hat{X}$ is deterministic, therefore $f^{-1}(\hat{X})$ and $Y$ are mutually independent conditioned on $\hat{X}$.
Then we conclude that (\ref{bottle}) is satisfied: 
$$I(Y,X) = I(Y,\hat{X}_{\text{ClofNet}}).$$
It implies no information is lost during the transformation from the original input to ClofNet's input.

\subsection{Universal approximation}\label{univer}
A crucial ingredient for proving the universal approximation property of a model is to apply Stone-Weierstrass theorem in a proper space. This is nontrivial if the function class we are considering has symmetry. For functions defined on a graph, we want them to be equivariant with respect to permutations on the nodes. Therefore, \citet{maehara2019simple} introduces the graph space as the space of equivariant classes of isomorphism graphs. Additional difficulty arises when SE(3)-symmetry is involved. By moving the graph's center, we only need to consider the SO(3)-symmetry acting on each node's feature globally. Now we discuss two graph network structures that can achieve universality. Both of them can be seen as a realization of equivariant polynomials neural networks.

\paragraph{Minimal universal architecture with equivariant frame} \label{app:mini}
In section 4 of \citep{dym2020universality}, the authors proposed an equivariant network that can achieve SE(3) equivariance by tensor representations. We briefly sketch the general approach of \citep{dym2020universality}.

Roughly speaking, a function class $\mathbf{F}_C(\Ffeat,\Fpool)$  with universal approximation property can be decomposed into two conditions:
\begin{itemize}
    \item \textbf{Condition 1.} $\Ffeat$ is able to represent all polynomials which are translation invariant and permutation-equivariant;
    \item \textbf{Condition 2.} $\Fpool$ contains all linear SO(3)-equivariant mappings between the range of $\Ffeat$:$\Wfeat$, and the output tensor space.
\end{itemize}
Furthermore, \citep{dym2020universality} showed that for a non-negative integer vector $\vec{r}=(r_1,\dots,r_k)$,  $Q^{(\vec{r})}=\{(Q_j^{(\vec{r})\}})_{j=1}^n $ defined by
\begin{equation}\label{Q}
Q_j^{(\vec{r})}(\mX)=\sum_{i_2,\ldots,i_K=1}^n \vx_j^{\otimes r_1}\otimes \vx_{i_2}^{\otimes r_2} \otimes  \vx_{i_3}^{\otimes r_3}\otimes \ldots \otimes  \vx_{i_K}^{\otimes r_K}.
\end{equation}
are SO(3)-equivariant and satisfy condition 1. \textbf{In fact, they are (r,0)-type tensors, and transform equivariantly by definition (\ref{ten})}. Here, each $x_i$ is already centralized and there is no need to centralize them as (6) of \citep{dym2020universality}. To express these tensor products, \citet{dym2020universality} designed a \textbf{parameterized minimal universal architecture neural layer} ((7) of \citep{dym2020universality}) from input $(\mX,\vv)$ to the output $(\mX,\Tilde{\vv})$:
$$\Tilde{\vv}_j(\mX, \vv | \theta_1 , \theta_2) = \theta_1 \vx_j \otimes \vv_j + \theta_2 \sum_i \vx_i \otimes \vv_i,$$
where $\vv_i$ is the input tensor feature of the body position $\vx_i$.
Now it's necessary to show how ClofNet performs direct tensor product under scalarization. We take a simple example of the inner product between two vectors: $\pmb{\chi}_1 = \chi_1^a \va + \chi_1^b \vb + \chi_1^c \vc$ and  $\pmb{\chi}_2 = \chi_2^a \va + \chi_2^b \vb + \chi_2^c \vc$, under equivariant frame $(\va,\vb,\vc)$. Then the tensor product gives:
\begin{align} \nonumber
 \pmb{\chi}_1 \otimes \pmb{\chi}_2 =&  \chi_1^a  \cdot \chi_2^a \va \otimes \va +  \chi_1^b  \cdot \chi_2^b \vb \otimes \vb + \chi_1^c  \cdot \chi_2^c \vc \otimes \vc\\\nonumber
 &+ \chi_1^a  \cdot \chi_2^b \va \otimes \vb + \chi_1^b  \cdot \chi_2^a \vb \otimes \va\\ \nonumber
& + \chi_1^a  \cdot \chi_2^c \va \otimes \vc + \chi_1^c  \cdot \chi_2^a \vc \otimes \va\\ \label{product}
 &+ \chi_1^b  \cdot \chi_2^c \vb \otimes \vc + \chi_1^c  \cdot \chi_2^b \vc \otimes \vb.
\end{align}
Under scalarization, the output of this tensor product is represented by the set containing second order polynomials of the invariant scalars ($\{\chi_1^a \dots \chi_2^c\}$) defined by (\ref{scalar}). 

Now we are ready to present the \textbf{parameterized minimal universal architecture neural layer with equivariant frames}. After scalarization, denote the invariant scalars of $\{\vx_i\}_{i=1}^n$ and $\{\vv_i\}_{i=1}^n$ by $S^j(\vx_i) := \{S^j(\vx_i)\}_{j=1}^3$ and $S^j(\vx_i) := \{S^j(\vv_i)\}_{j=1}^m$, where $j$ is the coefficients index of $\vv_i$. Then each layer maps $(\mX,\vv)$ to the output $(\mX,\Tilde{v})$:
\begin{equation} \label{stupid}
(\mX,\vv) \xrightarrow{\text{Scalarization}} (S(\mX),S(\vv)) \longrightarrow \Tilde{v} := \Tilde{v}(S(\mX),S(\vv) | \theta_1, \theta_2),\end{equation}
and $\Tilde{v}(S(\mX),S(\vv) | \theta_1, \theta_2)$ is given by
$$\Tilde{v}_i^{jk}(S(\mX),S(\vv) | \theta_1, \theta_2) = \theta_1 S^j(\vx_i) \cdot S^k(\vv_i) + \theta_2 \sum_l S^j(\vx_l) \cdot S^k(\vv_l),$$
where $jk$ denotes the coefficients index of $\Tilde{v}$ (each operation will bring an extra index for the features). Note that $\Tilde{v}$ is a tuple of invariant scalars, and it will connect to the vectorization block when the final output is a higher-order tensor. 
\begin{proof}[\textbf{Proof of theorem \ref{universal approximation}}]
Suppose our target is to predict an equivariant tensor associated with a pre-fixed node with input the many-body system $\mX=(\vx_1,\dots,\vx_N) \in \mathbb{R}^{N\times 3}$. Then we fix a single equivariant frame on the pre-fixed node by averaging our edgewise frames(bases) around this node. Under this frame, we \textbf{scalarize $\{\vx_i\}_{i=1}^n$ of all nodes (not only the 1-hop nodes)}.
To relate the minimal universal architecture with ClofNet, we emphasise that formula (\ref{product}) already shows that with the equivariant frame, \textbf{tensor products are ordinary multiplication of invariant scalars.} From a commutative diagram point of view, let $\vf{f}(\pmb{\chi}_1,\pmb{\chi}_2)$ denote the tensor product on two equivariant vectors $\pmb{\chi}_1 = \chi_1^a \va + \chi_1^b \vb + \chi_1^c \vc$ and  $\pmb{\chi}_2 = \chi_2^a \va + \chi_2^b \vb + \chi_2^c \vc$ under equivariant frame $(\va,\vb,\vc)$ and denote the scalarization of $\pmb{\chi}_1$ and $\pmb{\chi}_2$ by $s_1$ and $s_2$, then we have the following commutative diagram:
\begin{tikzcd} \label{commutative2}
     \pmb{\chi}_1, \pmb{\chi}_2 \arrow[r,  "\vf"] \arrow[d, "\text{Scalarize}"]
      & \mathbf{\theta}_1 \otimes \mathbf{\theta}_2  \\
     s_1,s_2 \arrow[r, "\Tilde{f}"]
      & \Tilde{f}(s_1,s_2) \arrow[u,  "\text{Vectorize}"'],
\end{tikzcd}
and $\text{Vecterize} \circ \Tilde{f}$ is exactly the right hand side of (\ref{product}), where $\Tilde(s_1,s_2)$ is a combination of second order polynomials of $s_1$, $s_2$. Here, the Vectorize block is under the second-order equivariant frame built from $(\va,\vb,\vc)$: $\{\va \otimes \va, \dots, \vb \otimes \vc \}$.
\footnote{\textbf{Recall from our construction of the frame, vector $\va$ denotes the radial direction. Then if we perform tensor product only along this direction, the results won't expand the whole space of the second-order tensors, the direction degeneration problem would become critical.}}The same story holds for tensor product of higher orders, and we conclude that an neural network which can approximate permutation-equivariant polynomials universally (not necessary (\ref{stupid}))  meets condition 1 in the ClofNet setting.

In terms of condition 2, \citet{dym2020universality} gave an explicit construction in the appendix of all equivariant linear functionals when the output is a scalar. It's not a coincidence that both the equivariant frame (\ref{basis}) and (23) of \citep{dym2020universality} used the inner product and determinant (equivalent to the cross product). However, with our equivariant frame, we can give a straightforward argument on building equivariant functionals.

By the commutative diagram (\ref{commutative2}), both the input and output tensor space are scalarized under the equivariant frame $(\va, \vb, \vc)$ (and the tensor product of equivariant frame). Furthermore, a \textbf{linear} functional $\phi$ is fully determined by its value on the frame. Suppose we are considering the map from order two tensors to order one tensors (vectors), then $\phi$ is determined by the coefficients (SO(3)-invariant) $(\phi^1,\phi^2,\phi^3)$:
$$\phi(\va \otimes \vb) = \phi^1(\va \otimes \vb) \va + \phi^2(\va \otimes \vb) \vb + \phi^3(\va \otimes \vb) \vc,$$
for all possible combinations $\va,\vb \in \{\va,\vb,\vc\}$. \textbf{Differentiate with the `cumbersome' method of constrcuting equivariant linear functionals in \citep{dym2020universality}, $\phi$ is a linear functional with no equivariance restriction, which is straightforward to construct with a fully-connected linear layer.} Combining the two parts, we have constructed the universal architecture with equivariant frame. 
\end{proof}
\begin{remark}[Incorporating geometric (physical) quantities as steerable features]\label{Incorporating}
The general definition of steerable features can be found in (8) of \citep{freeman1991design}. Briefly, frame coefficients of a function expanded in an equivariant frame (including sphere harmonics and tensor products of our equivariant frames) satisfy (8) of \citep{freeman1991design}. Then \citet{2110.02905} injects geometric quantities (forces, momentum, position vectors, ...) by performing CG tensor product between the chosen geometric quantities and each layers' output (also equivariant tensors), to make the model more expressive. However, the message from section 4 of \citep{dym2020universality} indicates that the computational expensive SO(3) representations (spherical harmonics and CG decomposition) can bring for the expressiveness, can be realized by direct tensor product (formula (\ref{tensor product bases}) is a special case of generating tensor product equivariant bases from our equivariant frame). As we have shown in the above proof, the direct tensor product under the equivariant frame can be realized by the multiplication of invariant scalars.

In conclusion, to add geometric quantities in a steerable way, we can first scalarize the steerable features, and then either encode the multiplication operation on invariant scalars and the steerable features explicitly or concatenate the steerable features with the invariant scalars and implement nonlinear transformations like MLP that can implicitly approximate the tensor product operation in ClofNet.
\end{remark}
\paragraph{Tensorized graph neural network with equivariant frame}
Tensorized graph neural network can be seen as another way of approximating permutation-equivariant polynomials on a graph, such that (\ref{universal approximation}) holds. Since the high level idea is the same, we neglect some details in this section.

Following (7) of \citep{maehara2019simple}, the graph is represented by
$$\mathcal{G}_0 = \{W \in \mathbb{R}^{n \times n} : |W(i,j)|\leq 1,\ \forall i,j\},$$
which can be seen as a graph with bounded scalar-valued edge features. To encode tensor field-valued edge features (all node features are identified with edge features) into this representation, we introduce two more indexes $1 \leq u,v \leq m$ for $W(i,j)$, such that
$$W(i,j,u,v) \in \mathbb{R}^2,\ \forall i,j,u,v,$$
where $u$ indicates the u-th component of the tensor field at the i-th node, and $v$ is defined similarly. In other words, $W(i,j,u,v)$ records the u-th component of the tensor field \textbf{scalarized by the equivariant frame} associated with the i-th node and the v-th component of the tensor field \textbf{scalarized by the equivariant frame} associated with the j-th node.  Then under the group action $g \in SO(3)$, SO(3) equivariance requires
$$W(g \cdot x_i, g \cdot x_j,\cdot,\cdot) = g W(x_i,x_j,\cdot,\cdot).$$
Therefore, we should identify $W_1$ and $W_2$ if
\begin{equation} \label{eq:equi}
W_1(g \cdot x^1_i, g \cdot x^1_j,\cdot,\cdot)  \equiv W_2(x^2_i,x^2_j,\cdot,\cdot),\end{equation}
for each $1 \leq i,j \leq n$.
We denote $W_1 \sim W_2$ if $W_1$ and $W_2$ are isomorphic and there exists $g \in SO(3)$ such that (\ref{eq:equi}) holds. In fact, since the group action on the 3D point cloud in a fiber-wise' way, the permutation of nodes $\sigma$ is commutative with the SO(3) group $g$ action:
$$W(\sigma(g \cdot x_i), \sigma(g \cdot x_j),\cdot,\cdot) = W(g \cdot \sigma(x_i), g \cdot \sigma(x_j),\cdot,\cdot).$$
We define the SO(3) graph space by $\Bar{\mathcal{G}}_0 = \mathcal{G}_0 / \sim$. The edit distance ((9) of \citep{maehara2019simple}) can be extended to include the additional indexes $u,v$. We further assume that $|W(i,j,u,v)| \leq 1$. Then the SO(3) graph space forms a metric space.

On the other hand, $g$ can also be seen as an orthogonal coordinate transformation from one base frame $(e_1,e_2,e_3)$ to another frame $(g \cdot e_1, g \cdot e_2, g \cdot e_3)$. Now, let $g_{ij}$ denote the coordinate transformation from the standard base frame to the equivariant frame of the edge $(i,j)$. Let
$$\Bar{W}_1(i,j,\cdot,\cdot) : = W_1(g_{ij} \cdot x^1_i, g_{ij} \cdot x^1_j,\cdot,\cdot).$$
Since the equivariant frame is intrinsically defined with respect to permutation of nodes, we have found a representative element $\Bar{W}_1$ for $W_1$ under the equivalence relation $\sim$ by combining all edges. \textbf{Note that similar representative elements can also be defined by projecting the features to the spherical harmonics space (TFN \citep{thomas2018tensor}), however, it's not sufficient to project the features into the radical direction only, like EGNN \citep{satorras2021nn}}.

Suppose we consider equivariant  vector field on graph $W$: $f: W \rightarrow \mathbb{R}^{3n}$, then it's easy to construct an edge-wise equivariant  vector field such that for each node $i$,
$$f_i = \sum_{j \in N(i)} f_{ij}.$$
Then under scalarization, each component of $f_{ij}$ becomes SO(3)-invariant scalar function. In this way, to prove the universality approximation property for equivariant vector field on graphs, it's sufficient to prove it on ordinary node-wise continuous functions $f: \Bar{\mathcal{G}}_0 \rightarrow \mathbb{R}$.

Let $\mathcal{F}$ be the set of simple (no loops) unweighted graphs with 1 labeled node, and let $F = (V(F), E(F)) \in \mathcal{F}$.
Let $G = (V(G), E(G); W)$ be a weighted graph, where $W \colon V(G) \times V(G) \to \mathbb{R}$ is the weighted adjacency matrix. for a given node $x \in V(G)$, we define the extended \emph{$1$-labeled homomorphism number} by (see \citep{LOVASZ2006962,maehara2019simple})
\begin{align}
\label{eq:labeled-homomorphism-density}
    \mathrm{hom}_{x}(F, W, G) 
    = \hspace{-2.0em} \sum_{\substack{\pi \colon V(F) \to V(G) \\ \pi(1) = x \ (i \in [k])}} & \prod_{i \in V(F)} W(\pi(i), \pi(i)) \nonumber \\
    \times \hspace{-1.0em} & \prod_{(i,j) \in E(F)}[\prod_{(u,v) \in G(i,j)}  W(\pi(i), \pi(j),u,v)],
\end{align}
where $G = \cup_{(i,j) \in E(G)} G(i,j)$, and $G(i,j)$ is a collection $(u,v)$ such that $u$ is a component of one tensor field on node $i$ and $v$ is a component of one tensor field on node $j$. From the definition, the extended  $1$-labeled homomorphism numbers are continuous node-wise functions in $W$ that is equivariant with respect to permutations of nodes. Now we introduce the universal approximation function class:
\begin{align}
    \mathcal{A} &= \left\{ W \mapsto \hspace{-0.3em} \sum_{F \in \mathcal{F}, G \in \mathcal{G}}^{\text{finite}} a_F \mathrm{hom}_{x}(F, W + 2 I,G) : a_F \in \mathbb{R} \right\},
\end{align}
where $I$ denotes the identity matrix of $n \times n$. By utilizing Stone-Weierstrass theorem, we prove the following theorem:
\begin{theorem}
\label{thm:universality-equivariance}
$\mathcal{A}$ is dense in the continuous scalar functions.
\end{theorem}
We need to check three conditions before applying Stone-Weierstrass theorem:
\begin{itemize}
   \item $\mathcal{A}$ forms a subalgebra inside continuous functions;
    \item $\mathcal{A}$ separates points, i.e., for any $x \neq y$, there exists $h \in \mathcal{A}$ such that $h(x) \neq h(y)$, and
    \item There exists $u \in \mathcal{A}$ that is bounded away from zero, i.e., $\inf_{x \in \mathcal{X}} |u(x)| > 0$.
\end{itemize}
\begin{lemma}
\label{lem:algebra-equivariant}
$\mathcal{A}'$ forms an algebra.
\end{lemma}
\begin{proof}
It's straightforward to check the closeness under the addition and the scalar multiplication.
It is closed under the product because of the following identity.
\begin{align}
    \mathrm{hom}_{x}(F_1, W,G_1) \mathrm{hom}_{x}(F_2, W,G_2) = \mathrm{hom}_{x}(F_1 \sqcup F_2, W, G_1 \sqcup G_2).
\end{align}
where $F_1 \sqcup F_2$ is the graph obtained from the disjoint union of $F_1$ and $F_2$ by gluing the labeled vertices.
\end{proof}

\begin{lemma}
\label{lem:bounded-away-from-zero-equivariant}
$\mathcal{A}$ contains an element that is bounded away from zero.
\end{lemma}
\begin{proof}
Let $\circ$ be the graph of 1 isolated vertices (singleton).
Then, $W \mapsto \mathrm{hom}_{x}(\circ, W + 2 I, 	\varnothing) \ge n - 1$ is bounded away from zero. 
\end{proof}
To prove the separate point property, we use the following theorem.
\begin{theorem}[Lemma 2.4 of {\citep{LOVASZ2006962}}]
\label{thm:lovasz2006rank}
Let $W_1, W_2 \in \mathbb{R}^{n \times n}$ be matrices with positive diagonal elements.
Let $x_1, x_2$ be two nodes of $W_1$ and $W_2$.
Then, $(W_1, x_1)$ and $(W_2, x_2)$ are isomorphic if and only if $\mathrm{hom}_{x_1}(F, W_1,\varnothing) = \mathrm{hom}_{x_2}(F, W_2,\varnothing)$ for all simple unweighted graph $F$.
\qed
\end{theorem}
Therefore, two non-isomorphic graphs are separable by $\mathcal{A}$. On the other hand, for two isomorphic graphs with different edge features, we can choose nonempty $G$ to separate them by the above theorem (since polynomials separate different edge features of a given edge).
\begin{lemma}
$\mathcal{A}$ separates points in $\Bar{\mathcal{G}}_0$.
\end{lemma}
The usual \textbf{tensorized graph neural network can express $E(F)$-fold scalar-valued tensor product}. However we can further extend it trivially to include vector-valued features on each edge by enlarging the index space's dimension and performing scalarization, and we denote it by the extended tensorized graph neural network. 
As it was demonstrated in \citep{maehara2019simple}, the 1-labeled homomorphism numbers can be implemented in a tensorized graph neural network \citep{keriven2019universal}. Then, the extended 1-labeled homomorphism numbers can be implemented in an extended tensorized graph neural network with equivariant bases in the same way. Combining the above, we have proved the universality property of  tensorized graph neural network equipped with equivariant bases:
\begin{theorem} \label{appro}
The tensorized graph neural network equipped with equivariant bases has universality in continuous equivariant functions.
\qed
\end{theorem}    

\subsection{Additional Experiments}
\subsubsection{Newtonian many-body system}
\label{app:newton}

\paragraph{Dataset}
For the Newtonian many-body system experiment, we follow EGNN \citep{satorras2021en} to generate the trajectories for four systems. The original source code for generating trajectories comes from \citep{kipf2018neural} (\url{https://github.com/ethanfetaya/NRI}) and is modified by EGNN (\url{https://github.com/vgsatorras/egnn}). We further extend the version of EGNN to three new settings, as described in Section \ref{draft:newton}. Similar to EGNN, we generate 5,000 timesteps for each trajectory and slice them from 3,000 to 4,000 to move away from the transient phase.
For the second experiment, we generate a new training dataset with 50,000 trajectories for the G+ES system, following the same procedure. The validation and testing datasets remain unchanged from the first experiment. 
We provide the physical evolution equations mentioned in Section \ref{draft:newton} as follows.

\textbf{ES. } This system consists of $N$ particles controlled by the electrostatic field, i.e., for each trajectory we are provided with initial positions $\mX(0) \in \mathbb{R}^{N \times 3}$, velocities $\mV(0) \in \mathbb{R}^{N \times 3}$ and charges $\{q_1, \cdots q_N\} \in \{-1, 1\}^N$. The time evolution of the particles is given by
\begin{equation}
\ddot{\vx}_i = \sum_{j \neq i}  q_iq_j \frac{\vx_i - \vx_j}{\norm{\vx_i - \vx_j}^3},\ \ \ \ 1 \leq i \leq N.
\end{equation}

\textbf{G+ES. } This system consists of $N$ particles controlled by both the electrostatic field and  external gravity force of the form: $\vf_{g}=(0,\, 0,\, g)$. The time evolution of the particles are given by
\begin{equation}
\ddot{\vx}_i = \sum_{j \neq i}  q_iq_j \frac{\vx_i - \vx_j}{\norm{\vx_i - \vx_j}^3} + \vf_{g},\ \ \ \ 1 \leq i \leq N.
\end{equation}
\textbf{L+ES. } This system consists of N particles controlled by both the electrostatic field and a Lorentz-like force field, which means there exists a force perpendicular to the direction of velocity $\vv$, i.e., $\vf_{l}(\vv) = q\vv \times \mathbf{B}$, where $q$ and $\mathbf{B}$ denote the charge of particles and the direction vector of the electromagnetic field respectively. 
The time evolution of the particles is given by:
\begin{equation}
\ddot{\vx}_i = \sum_{j \neq i} {[q_iq_j \frac{\vx_i - \vx_j}{\norm{\vx_i - \vx_j}^3} + \vf_{l}^i(\vv_i)]},\ \ \ \ 1 \leq i \leq N.
\end{equation}
\paragraph{Models} We implement all baselines (GNN, Radial Field, TFN and SE(3)-Transformer) by referring the codebase of EGNN. The architecture parameters (e.g., feature dimension, activation function) are adapted from EGNN. The learning rate and training epochs are tuned independently for each model. 

\subsubsection{Partially observed N-body experiment} \label{partial}
The N-body dynamical systems considered in the main text can be classified into the category of Markov process, where the system's immediate future state $\mX(t+\Delta t)$ (future positions and velocities) is fully determined by its current state $\mX(t)$ (current positions and velocities).In other words, let $\Delta t \rightarrow 0$, there exists $f$ such that $$\dot{\mX}(t) = f(\mX(t))$$
where $\dot{\mX}(t)$ denotes the gradient of $\mX(t)$,
and $f$ is independent of time $t$. 

Now we define a harder trajectory prediction task of a non-Markov dynamical system \textbf{POS}, such that $\dot{\mX}(t) = f(\mX(t),t)$ to further test ClofNet. \textbf{POS} consists of six particles under Newton's gravitation force, but only four of them could be observed, i.e. for the whole trajectory, we are provided with positions $\mX(t) \in \mathbb{R}^{4 \times 3}$ and velocities $\mV(t) \in \mathbb{R}^{4 \times 3}$ of \textbf{four} particles. Since another two unobserved particles are hidden, the sub-system of the given four particles is non-Markov. The time evolution of the particles is given by
\begin{equation}
\ddot{\vx}_i = \sum_{j \neq i}  -m_j \frac{\vx_i - \vx_j}{\norm{\vx_i - \vx_j}^3},\ \ \ \ 1 \leq i \leq 4.
\end{equation}

\textbf{Problem definition.} In this experiment, we apply our model to predict the long-term motion trajectory of \textbf{POS} given its initial ($t=0$) position and velocity. Following \citep{zhuang2020adaptive,li2021machine}, we formulate the trajectory prediction as two tasks: \textbf{Interpolation} and \textbf{Extrapolation}. The experimental setting is as follows. To generate the trajectory given the initial condition, we use observations $\vx_i(t), t \in \{\Delta t, \, 2\Delta t \dots , T_1\}$ as the training labels and observations $\vx_i(t), t \in \{T_1+\Delta t,\, T_1+\Delta t, \dots , T_2\}$ as the validation set. To evaluate the interpolation and extrapolation capacity of all methods, the observations $\vx_i(t), t \in \{\frac{1}{2}\Delta t,\, \frac{3}{2}\Delta t, \dots , T_1+\frac{1}{2}\Delta t\}$ and $\vx_i(t), t \in \{T_2+\Delta t,\, T_2+2\Delta t, \dots , T_3\}$ are used as \textbf{interpolation} and \textbf{extrapolation} test sets respectively. We measure the mean square error (MSE) between the predicted trajectory and ground truth. To measure the exactness of equivariance, we follow \citep{fuchs2020se} to apply uniformly sampled SO(3)-transformations on the input and output. The MSE between the predicted trajectory with rotated input and rotated ground truth could reflect the transformation robustness of method. The normalized distance between the rotated prediction with original input and the original prediction with the rotated input defines the equivariance error $\Delta_{EQ}$:
\begin{equation}
    \Delta_{EQ} = \norm {L_s \Phi (\vx) - \Phi L_s(\vx)} / \norm{L_s \Phi (\vx)},
\end{equation}
where $L_s$ and $\Phi$ denote SO(3) transformations and equivariant neural networks, respectively.

\textbf{Learning Framework and Implementation Details.} Inspired by \citep{norcliffe2020second}, for a Newtonian system, we utilize ClofNet to parameterize its acceleration vector field and adopt numerical ODE solver to integrate both the position and velocity trajectories. Only the MSE between predicted position trajectory and ground truth is taken as the loss penalty:
\begin{equation} 
\mathcal{L}(\theta) = \frac{1}{n}\sum_{i=1}^n L_2(x_{t_i}, \text{ODE}(\vx_{t_0},\vv_{t_0}, t_0, t_i, \Theta)),\ \ t_0 < t_1 < \cdots < t_n,
\end{equation}
where $(\vx_{t_0},\vv_{t_0})$ and $\Theta$ denote the initial condition of the system and the parameters of ClofNet $\Phi$. We compare our method to another two efficient equivariant models designed for vector field modeling: Radial Field \citep{kohler2019equivariant} and EGNN \citep{satorras2021nn,satorras2021en}.
Following \citep{zhuang2020adaptive}, the trajectory is simulated using the \textit{Dopri5} solver \citep{dormand1980family} with the tolerance to $10^{-7}$ and the modified physical rules.
The trajectory points are uniformly sampled with $\Delta t=10^{-3}$.
We implement all baselines and our method with Pytorch \citep{paszke2019pytorch}. All models use the same ODE solver (\textit{Dopri5}) as the evolving blocks and are trained with Adam optimizer \citep{kingma2014adam} via an MSE loss for 800 epochs. We set the number of layers to $2$ for all models and adjust the hidden dimensions of each model separately to keep the parameters in the same level. We adopt EGNN from \citep{satorras2021nn} for outputting vectors. 

\begin{table}[t]
  \caption{MSE of the \textbf{POS} dataset. \textit{Inter.} and \textit{Extra.} denote the interpolation and extrapolation task respectively.}
  \label{table:pos}
  \centering
  \begin{tabular}{c|cccc}
    \toprule
      Setting & Method & Inter. & Extra. & $\Delta_{EQ}$\\
    \midrule
    \multirow{3}{*}{POS} 
    & Radial Field & $1.996$ & $7.665$ & $5.77 \cdot 10^{-6}$ \\
    & EGNN & $0.726$ & $6.449$& $9.19\cdot 10^{-7}$ \\
    & ClofNet & $\textbf{0.138}$ & $\textbf{2.502}$& $4.13\cdot 10^{-6}$ \\
    \bottomrule
  \end{tabular}
\end{table}
\textbf{Results.}
From Table \ref{table:pos}, we have the following conclusions: (1) ClofNet outperforms all other default equivariant methods in the interpolation and extrapolation tasks, which demonstrates that ClofNet exhibits stronger expressive power by representing geometric information losslessly; (2)The equivariance error $\Delta_{EQ}$ of the three models are small (Due to the existence of numerical errors, $\Delta_{EQ}$ cannot be strictly zero), which empirically demonstrate the equivariance of these models.

\subsubsection{Molecular conformation generation}
\label{app:mcg}
\paragraph{Dataset}
For each dataset, $40,000$ molecules are randomly drawn and $5$ most likely conformations (sorted by energy) are selected for each molecule, and $200$ molecules are drawn from the remaining data, which results in $200,000$ conformations in the training set, $22,408$ and $14,324$ conformations in the test set for GEOM-QM9 and GEOM-Drugs datasets, respectively. The distances over distributions task are evaluated on the ISO17 dataset, where we follow the setup in \citep{simm2019generative}.
\paragraph{Learning Framework} 
Following \citep{shi2021learning}, for this first-order statistical ensemble system, we leverage a score-based generative modeling framework to estimate the gradient field of atomic positions (See more details about score-based networks in \citep{shi2021learning,song2020score} or Appendix \ref{evo}). A detailed illustration on the \textbf{equivariance} of the gradient fields (the score function) is also given in Appendix \ref{evo}.  The optimization objective of ClofNet $\Phi$ can be summarized as:
\begin{align} 
    \label{score}
    \small
    \mathcal{L}(\theta) = &\frac{1}{n}\sum_{i=1}^n \mathbb{E}_{\mX(0)}\{\lambda(t_i) \norm{\nabla H_{t_i} (\mX(t_i))\nonumber\\ &- \Phi(\mX(t_i),t_i,\Theta)}_2^2, 
     t_0 < t_1 < \cdots < t_n,
\end{align}

where $\lambda(t): [0,\,T] \rightarrow \mathbb{R}^{+}$ is a positive weighting function and $\nabla H_{t_i}$ is the pre-computed gradient field of noisy atomic positions (see Appendix \ref{evo}). Once the score network is optimized, we can use an annealed Langevin dynamics (ALD) sampler 
to generate conformations \citep{song2019generative,song2020score}.
\paragraph{Implementation Details} Besides the geometric input, we feed the node type, edge type and relative distances as extra node/edge attributes into the graph transformer block. 
Our score-based training framework keeps the same as \citep{shi2021learning}. The maximum and minimum noise scales are set to $10$ and $0.01$. Let $\{\sigma_i\}_{i=1}^L$ be a positive geometric progression scheme with a common ratio, we split the noise range into 50 levels.
For the reverse process, we follow ConfGF to use the annealed Langevin dynamics sampling method to generate stable structures. The sample step size $\eta_s$ is chosen according to \citep{song2020improved}. All hyper-parameters mentioned in the forward and reverse process are kept the same as \citep{shi2021learning}.
The results reported in Table \ref{table:molecular} are copied from \citep{shi2021learning} considering that we rigorously evaluate ClofNet on the same benchmark and data split setting. 

\paragraph{Conformation Generation}
Here we introduce the calculation equation of RMSD:
\begin{equation}
\text{RMSD}(R,\hat{R}) = \min(\frac{1}{n}\sum_{i=1}^{n}||R_i-\hat{R_i}||^2)^{\frac{1}{2}},
\end{equation}
where $n$ denotes the number of heavy atoms.

We visualize several conformations in the Drugs dataset in Figure \ref{fig:visual} that are best aligned with the reference ones generated by different methods, illustrating ClofNet's superior capacity on generating high-quality drug molecular conformation. 
\begin{figure*}[t]
    \centering
    \includegraphics[width=0.9\linewidth]{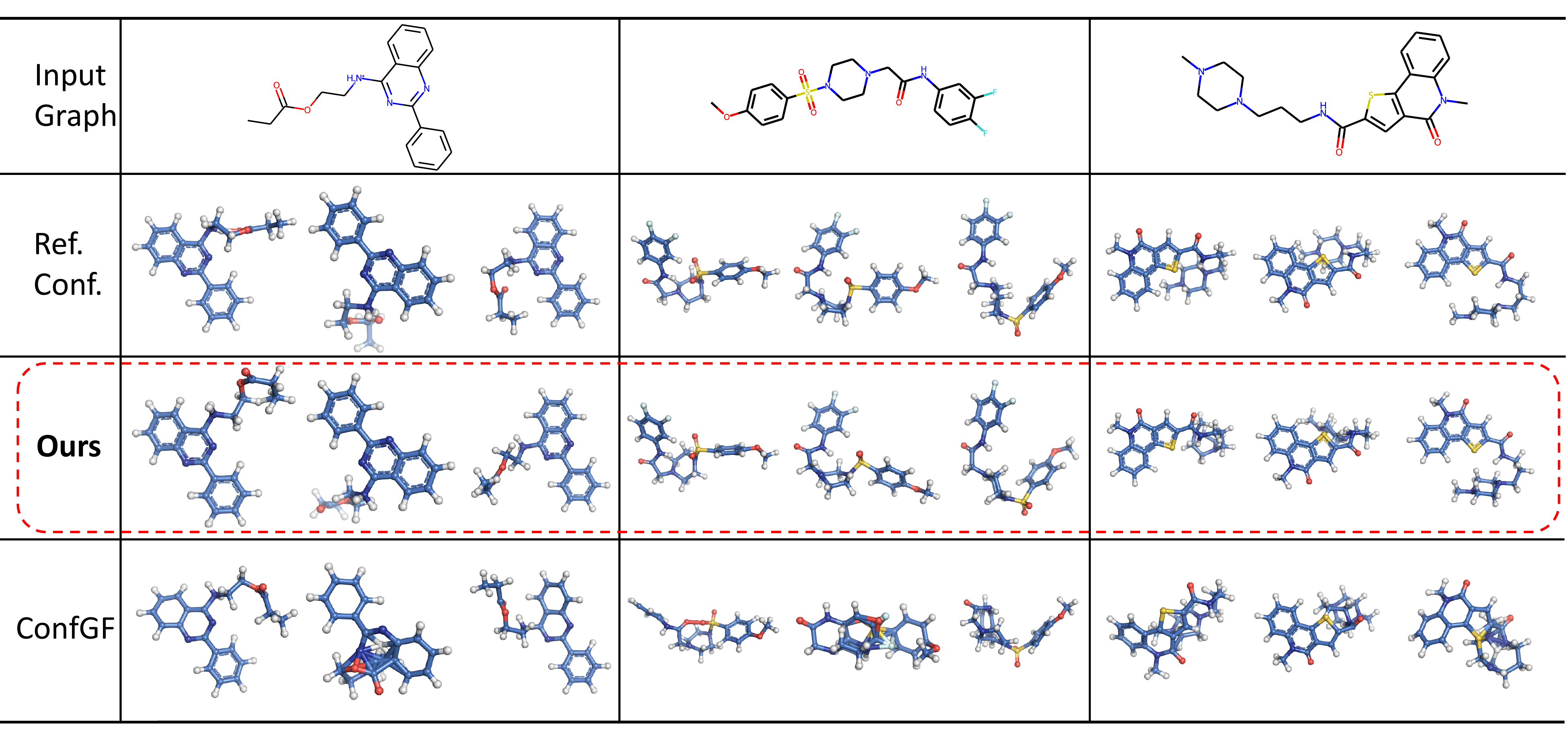}
    \caption{Visualizations of generated conformations. For each molecule randomly selected from GEOM-Drugs dataset, we sample multiple conformations and show the best-aligned ones with the reference ones.}
    \label{fig:visual}
\end{figure*}



\end{document}